\documentclass{article}
\usepackage{fullpage}

\usepackage{epsfig,graphicx,graphics,color}
\usepackage{amsmath,amsthm,amssymb,color,comment,url}
\usepackage{mathtools,thmtools}
\usepackage{thm-restate}
\usepackage{caption}
\usepackage[position=b]{subcaption}
\usepackage[unicode]{hyperref}
\hypersetup{
    colorlinks=true,       %
    linkcolor=blue,        %
    citecolor=red,         %
    filecolor=magenta,     %
    urlcolor=cyan,         %
    linktocpage=true
}
\usepackage[sort,nocompress]{cite}
\usepackage{titling} %
\thanksmarkseries{alph}
\usepackage{enumerate}

\def\final{0}  %
\def\iflong{\iffalse}
\ifnum\final=0  %
\usepackage[dvipsnames]{xcolor}
\newcommand{\yu}[1]{{\color{red}[{\tiny \textbf{Yu:} \bf #1}]\marginpar{\color{red}*}}}
\newcommand{\yutaro}[1]{{\color{red}[{\tiny \textbf{Yutaro:} \bf #1}]\marginpar{\color{red}*}}}
\newcommand{\tamas}[1]{{\color{red}[{\tiny \textbf{Tam\'as:} \bf #1}]\marginpar{\color{red}*}}}
\newcommand{\kristof}[1]{{\color{red}[{\tiny \textbf{Krist\'of:} \bf #1}]\marginpar{\color{red}*}}}

\else %
\newcommand{\yu}[1]{}
\newcommand{\yutaro}[1]{}
\newcommand{\tamas}[1]{}
\newcommand{\kristof}[1]{}
\fi

\numberwithin{equation}{section}

\theoremstyle{plain}
\newtheorem{theorem}{Theorem}[section]
\newtheorem{lemma}[theorem]{Lemma}

\newtheorem{claim}[theorem]{Claim}
\newtheorem{corollary}[theorem]{Corollary}

\theoremstyle{definition}
\newtheorem{definition}[theorem]{Definition}
\newtheorem{example}[theorem]{Example}

\newtheorem{remark}[theorem]{Remark}

\newcommand{\RR}{{\mathbb{R}}}
\newcommand{\ZZ}{{\mathbb{Z}}}
\newcommand{\bM}{{\mathbf{M}}}

\newcommand{\cF}{{\mathcal{F}}}

\newcommand{\cI}{{\mathcal{I}}}

\newcommand{\cImink}{{\mathcal{I}_\mathrm{min}^k}}
\newcommand{\cIminkk}{{\mathcal{I}_\mathrm{min}^{k+1}}}
\newcommand{\cl}{\mathrm{cl}}
\newcommand{\rmin}{{r_\mathrm{min}}}
\newcommand{\Amin}{{A^\mathrm{min}}}
\newcommand{\Amino}{{A_1^\mathrm{min}}}
\newcommand{\Amint}{{A_2^\mathrm{min}}}
\newcommand{\Amini}{{A_i^\mathrm{min}}}

\newcommand{\Asure}{{A^\mathrm{sure}}}
\newcommand{\Asuret}{{A_2^\mathrm{sure}}}

\newcommand{\Asusp}{{A^\mathrm{susp}}}
\newcommand{\Asuspo}{{A_1^\mathrm{susp}}}
\newcommand{\Asuspt}{{A_2^\mathrm{susp}}}

\newcommand{\Dmin}{{D^\mathrm{min}}}
\newcommand{\tA}{{\tilde{A}}}
\newcommand{\tD}{{\tilde{D}}}
\newcommand{\tP}{{\tilde{P}}}
\newcommand{\tQ}{{\tilde{Q}}}
\newcommand{\cC}{{\mathcal{C}}}
\newcommand{\true}{\text{\sc True}}
\newcommand{\false}{\text{\sc False}}

\DeclarePairedDelimiter{\set}{\{}{\}}
\DeclarePairedDelimiterX{\Set}[2]{\{}{\}}{\,#1\mathclose{}\nonscript\;\delimsize|\nonscript\;\mathopen{}#2\,}
\DeclarePairedDelimiter{\abs}{|}{|}
\newcommand{\R}{\mathbb{R}}
\newcommand{\Z}{\mathbb{Z}}
\newcommand{\Zp}{\Z_{\ge 0}}
\newcommand{\Rp}{\R_{\ge 0}}

\newcommand{\ones}{\mathbf{1}}
\newcommand{\rhomin}{\rho_{\mathrm{min}}}

\makeatletter
\def\namedlabel#1#2{\begingroup
    #2%
    \def\@currentlabel{#2}%
    \phantomsection\label{#1}\endgroup
}
\makeatother

\usepackage[algo2e,boxruled,linesnumbered,procnumbered]{algorithm2e}
\makeatletter
\renewcommand{\algocf@caption@boxruled}{%
  \hrule
  \hbox to \hsize{%
    \vrule\hskip-0.4pt
    \vbox{   
       \vskip\interspacetitleboxruled%
       \unhbox\algocf@capbox\hfill
       \vskip\interspacetitleboxruled
       }%
     \hskip-0.4pt\vrule%
   }\nointerlineskip%
}%
\makeatother

\usepackage{tcolorbox}

\newtcolorbox{probbox}{arc=6pt,
                      colback=white!100,
                      colframe=black!50,
                      before skip=6pt,
                      after skip=6pt,
                      boxsep=1pt,
                      left=6pt,
                      right=6pt,
                      top=4pt,
                      bottom=4pt}

\newcommand{\searchprob}[3]{
   \begin{center}%
    \begin{minipage}{0.96\linewidth}%
      \begin{probbox}
      \textsc{#1}\\[0.2ex]
      \textbf{Input:} #2\\[0.2ex]
      \textbf{Goal:} #3
      \end{probbox}
    \end{minipage}%
  \end{center}
}

\usepackage{etoolbox}
\newcounter{algoline}

\AtBeginEnvironment{algorithm}{\setcounter{algoline}{0}}

\title{Matroid Intersection under Minimum Rank Oracle}
\thanksmarkseries{alph}
\author{
Mih\'aly B\'ar\'asz\thanks{HUN-REN-ELTE Egerv\'ary Research Group, Department of Operations Research, E\"otv\"os Loránd University, Budapest, Hungary. Email: \texttt{mihaly.barasz@gmail.com}.}
\and
Kristóf Bérczi\thanks{MTA-ELTE Matroid Optimization Research Group and HUN-REN–ELTE Egerváry Research Group, Department of Operations Research, ELTE Eötvös Loránd University, and HUN-REN Alfréd Rényi Institute of Mathematics, Budapest, Hungary. Email: \texttt{kristof.berczi@ttk.elte.hu}.}
\and
Tamás Király\thanks{HUN-REN-ELTE Egerv\'ary Research Group, Department of Operations Research, E\"otv\"os Loránd University, Budapest, Hungary. Email: \texttt{tamas.kiraly@ttk.elte.hu}.}
\and
Taihei Oki\thanks{Institution for Chemical Reaction Design and Discovery (WPI-ICReDD), Institution for Integrated Innovations, Hokkaido University, Sapporo, Hokkaido, Japan. Center for Advanced Intelligence Project, RIKEN, Tokyo, Japan. Email: \texttt{oki@icredd.hokudai.ac.jp}}
\and
Yutaro Yamaguchi\thanks{Department of Information and Physical Sciences, Graduate School of Information Science and Technology, Osaka University, Osaka, Japan. Email: \texttt{yutaro.yamaguchi@ist.osaka-u.ac.jp}.}
\and
Yu Yokoi\thanks{Department of Mathematical and Computing Science, School of Computing, Institute of Science Tokyo, Tokyo, Japan. Email: \texttt{yokoi@comp.isct.ac.jp}.}
}
\date{\empty}

\begin{document}
\maketitle

\begin{abstract}
In this paper, we consider the tractability of the matroid intersection problem under the minimum rank oracle. In this model, we are given an oracle that takes as its input a set of elements and returns as its output the minimum of the ranks of the given set in the two matroids.
For the unweighted matroid intersection problem, we show how to construct a necessary part of the exchangeability graph, which enables us to emulate the standard augmenting path algorithm.
For the weighted problem, the tractability is open in general.
Nevertheless, we describe several special cases where tractability can be achieved, and we discuss potential approaches and the challenges encountered.

On the positive side, we present a solution for the case where no circuit of one matroid is contained within a circuit of the other. 
Additionally, we propose a fixed-parameter tractable algorithm, parameterized by the maximum size of a circuit of one matroid.
We also show that a lexicographically maximal common independent set can be found by the same approach, which leads to a nontrivial approximation ratio for finding a maximum-weight common independent set.
On the negative side, we prove that the approach employed for the tractable cases above involves an NP-hard problem in the general case.
We also show that if we consider the generalization to polymatroid intersection, even the unweighted problem is hard under the minimum rank oracle.

\bigskip

\end{abstract}

\thispagestyle{empty}
\newpage
\pagenumbering{roman}
\tableofcontents
\thispagestyle{empty}
\newpage
\pagenumbering{arabic}
\setcounter{page}{1}

\section{Introduction}\label{sec:introduction}
When designing matroid algorithms, it is essential to clarify how matroids are represented. Since the number of bases is potentially exponential, listing them explicitly is impractical. Instead, standard oracles are assumed to be available to query the rank or the independence of a subset, and the complexity of the algorithm is measured by the number of oracle calls. In~\cite{berczi2023matroid}, B\'erczi, Kir\'aly, Yamaguchi, and Yokoi initiated the study of matroid intersection problems under restricted oracles, with particular emphasis on the rank sum, common independence, and maximum rank oracles. They showed that the weighted matroid intersection problem remains tractable even under the rank sum oracle. Furthermore, they showed that the common independence oracle model allows for an efficient algorithm for the unweighted matroid intersection problem when one of the matroids is a partition matroid, and that even the weighted case is solvable when one of the matroids is an elementary split matroid, which is a generalization of a paving matroid~\cite{berczi2023hypergraph}. 

As a continuation of the work~\cite{berczi2023matroid}, we investigate the tractability of the matroid intersection problem under the minimum rank oracle. In this model, for a subset of elements, the oracle outputs the minimum of the ranks of the subset in the two matroids.
In addition to purely theoretical interest, the motivation of the problem comes from a polyhedral perspective. In \cite{edmonds1970submodular}, Edmonds showed that the convex hull of the common independent sets of two matroids coincides with the intersection of the two independent set polytopes. This implies that the common independent set polytope itself is determined by the \emph{minimum rank function} of the two matroids rather than the two rank functions.
However, while the separation problem for an independent set polytope defined by a rank function reduces to submodular function minimization and therefore is tractable \cite{iwata2001,mccormick2005,schrijver2000}, this does not hold for a common independent set polytope defined by the minimum of the two rank functions, which is no longer submodular. In fact, one of our hardness results shows that separation for the polytope defined by the minimum of two \emph{polymatroid functions} requires an exponential number of oracle calls.

A set function is called \emph{$1/3$-submodular} and \emph{$2/3$-submodular} if for every distinct three subsets $X, Y, Z$, the submodular inequality holds for at least one and at least two, respectively, among the three pairs $\{X, Y\}$, $\{Y, Z\}$, and $\{Z, X\}$.
Mizutani and Yoshida~\cite{mizutani2024polynomial} gave polynomial-time algorithms to minimize 2/3-submodular functions. However, it was shown by B\'erczi and Frank~\cite{berczi2008variations} that minimizing a 1/3-submodular function requires an exponential number of oracle calls when the function is given by an evaluation oracle. This naturally raises the following question: apart from 2/3-submodular functions, what other classes of 1/3-submodular functions can be minimized in polynomial time?
Our problem concerns the class of functions that arise as the sum of a modular function and the minimum rank of two matroids. On the one hand, these functions are 1/3-submodular, since at least two of every three subsets attain the minimum rank in the same matroid. On the other hand, the minimization problem over this class is equivalent to the separation problem over the common independent set polytope given by minimum rank function of two matroids; with the aid of the ellipsoid method~\cite{grotschel1988geometric}, the latter problem is in turn equivalent to the weighted matroid intersection problem under the minimum rank oracle, which we focus on this paper.

There is a sort of common sense in combinatorial optimization that the tractability of a problem usually can be explained from a polyhedral point of view (cf.~the title of Schrijver's book~\cite{schrijver2003combinatorial}).
From this perspective, it is valuable to clarify what is the minimum amount of information to make the optimization/separation problem on the matroid intersection polytope tractable as it is completely determined by the minimum rank function.
On the other hand, hardness results on matroid problems are often obtained through constructions based on paving matroids (e.g., \cite{jensen1982complexity, lovasz1981matroid, berczi2021complexity}), but B\'erczi et al.~\cite{berczi2023matroid} revealed that the weighted matroid intersection problem is efficiently solvable if one of the matroids is paving, even when the oracle is restricted to the common independence oracle, which is weaker than our minimum rank oracle.
Thus, settling the computational complexity of the weighted matroid intersection problem under the minimum rank oracle would lead to significant progress regardless of whether it is tractable or hard: proof of tractability would result in a new kind of matroid intersection algorithm, while proof of hardness would require a novel approach to computational difficulty in matroids.

\subsection*{Our results}
As in the rank sum case of~\cite{berczi2023matroid}, the main obstacle to giving an efficient algorithm is that the usual augmenting path approach cannot be applied, since the exchangeability graphs are not determined by the minimum rank oracle. However, our aim is still to emulate the breadth first search (and the Bellman--Ford algorithm) for finding a shortest (and cheapest, respectively) augmenting path on the underlying exchangeability graph, but our strategy is completely different from the one in~\cite{berczi2023matroid}.

The main difficulty caused by the minimum rank oracle is that we cannot determine for each element its fundamental circuits in the two matroids separately.
To overcome this difficulty, we slightly modify the exchangeability graph by roughly estimating the fundamental circuits so that extra arcs that in reality are not present.
For the unweighted problem, we can do it only using the minimum rank oracle with preserving the set of shortest augmenting paths, which in turn allows us to solve the unweighted problem (Theorem~\ref{thm:unweighted}).
We note that this fact has appeared as a preprint~\cite{egresqp-06-03}, and we present in a revised form to make attempt of its extension.

In order to tackle the weighted problem, by further refining the notion of modified exchangeability graphs, we introduce the notion of consistent exchangeability graphs that have strong structural properties.
Specifically, we show that even if the true exchangeability graph is not correctly recognized, a kind of consistency with respect to small local exchanges is enough for the graph to behave essentially the same (Lemma~\ref{lem:correctness}), which is sufficient to emulate the usual weighted matroid intersection algorithm.
While it unfortunately turns out difficult to find such a graph in general (Theorem~\ref{thm:NP-hard}), we show that an ``almost consistent'' exchangeability graph can be found efficiently by solving the 2-SAT problem (Lemma~\ref{lem:2-SAT}).
This leads to the tractability under the minimum rank oracle for several special cases: when no circuit of one matroid contains a circuit of the other (Theorem~\ref{thm:no_circuit_inclusion}), when the maximum size of a circuit of one matroid is bounded by a constant (Theorem~\ref{thm:FPT_circuit}), and when the objective is to find a lexicographically maximal common independent set (Theorem~\ref{thm:lexicographically_maximal});
the last one also implies at least $\min\{1, \alpha/2\}$-approximation for the usual weight maximization problem (Corollary~\ref{cor:approximation}), where $\alpha > 1$ is the minimum ratio of two distinct positive weight values.

We also consider a natural generalization of the problem to the polymatroid intersection setting, in which testing the feasibility is already hard (Theorem~\ref{thm:polymatroid}).

\subsection*{Organization}
The rest of the paper is organized as follows.
Basic definitions and notation are introduced in Section~\ref{sec:preliminaries}.
In Section~\ref{sec:unweighted_min_rank}, we introduce the concept of modified exchangeability graphs and show their crucial property, which leads to a solution for the unweighted problem.
In Section~\ref{sec:consistent}, we refine the modified exchangeability graphs further, leading to the definition of consistent exchangeability graphs. We show that solving the weighted matroid intersection problem reduces to determining a consistent exchangeability graph.
On the positive side, in Section~\ref{sec:2-SAT}, we show that an almost consistent exchangeability graph can be constructed in polynomial time by solving the 2-SAT problem, which is sufficient for solving several special cases.
On the negative side, we prove in Section~\ref{sec:NP-hard} that the problem of finding a consistent exchangeability graph is NP-hard in general.
Finally, in Section~\ref{sec:polymatroid}, we show that if we consider the generalization to polymatroid intersection, even the feasibility problem is hard under the minimum rank oracle.

\section{Preliminaries}
\label{sec:preliminaries}
For the basics on matroids and the matroid intersection problem, we refer the reader to \cite{oxley2011matroid, schrijver2003combinatorial,frank2011connections}.
To make the paper self-contained, we repeat the basic definitions and known facts on matroids that are summarized in~\cite[Section~2]{berczi2023matroid}.

Throughout the paper, for $i=1,2$, let $\bM_i=(E,\cI_i)$ be loopless matroids on the same finite ground set $E$ of size $n$, whose \textbf{independent set families}, \textbf{rank functions}, and \textbf{closure operators} are denoted by $\cI_i$, by $r_i$, and by $\cl_i$, respectively.
For $I \in \cI_i$ and $x \in \cl_i(I) \setminus I$, the \textbf{fundamental circuit} of $x$ with respect to $I$ in $\bM_i$ is $C_i(I, x) = \{\, y \in I + x \mid I + x - y \in \cI_i \,\}$. 
For two sets $X,Y\subseteq E$, we denote their \textbf{symmetric difference} by $X\triangle Y=(X\setminus Y)\cup(Y\setminus X)$.

Let us first overview some basic results on matroid intersection.
In \cite{edmonds1970submodular}, Edmonds gave the following characterization for the maximum cardinality of a common independent set of two matroids.

\begin{theorem}[Edmonds~\cite{edmonds1970submodular}]\label{thm:Edmonds}
The maximum cardinality of a common independent set of $\bM_1$ and $\bM_2$ is equal to
\begin{align*}
\min\left\{\, r_1(Z) + r_2(E \setminus Z) \mid Z \subseteq E \,\right\}.
\end{align*}
\end{theorem}

The notion of exchangeability graphs plays a central role in any matroid intersection algorithm.

\begin{definition}[Exchangeability Graphs]\label{def:exchange}
Assume that $I\in\cI_1\cap \cI_2$ is a common independent set of $\bM_1$ and $\bM_2$. The \textbf{exchangeability graph} with respect to $I$ is a directed bipartite graph $D[I] = (E \setminus I, I; A[I])$ defined as follows.
Set
\begin{align*}
  S_I &\coloneqq \{\, s \in E \setminus I \mid I + s \in \cI_1 \,\},\\
  T_I &\coloneqq \{\, t \in E \setminus I \mid I + t \in \cI_2 \,\},
\end{align*}
where elements in $S_I$ and in $T_I$ are called \textbf{sources} and \textbf{sinks}, respectively.
We then define $A[I] \coloneqq A_1[I] \cup A_2[I]$, where
\begin{align*}
  A_1[I] \coloneqq&\ \{\, (y, x) \mid x \in E \setminus I,~y \in I,~I + x - y \in \cI_1 \,\}\\
  =&\ \{\, (y, s) \mid s \in S_I,~y \in I \,\} \cup \{\, (y, x) \mid x \in E \setminus (I \cup S_I),~y \in C_1(I, x) - x \,\},\\[1mm]
  A_2[I] \coloneqq&\ \{\, (x, y) \mid x \in E \setminus I,~y \in I,~I + x - y \in \cI_2 \,\}\\
  =&\ \{\, (t, y) \mid t \in T_I,~y \in I \,\} \cup \{\, (x, y) \mid x \in E \setminus (I \cup T_I),~y \in C_2(I, x) - x \,\}.
\end{align*}
Note that $S_I$ and $A_1[I]$ depend only on $\cI_1$, and $T_I$ and $A_2[I]$ depend only on $\cI_2$.
\end{definition}

Brualdi \cite{brualdi1969comments} observed that the set $A_i[I]$ of exchangeability arcs satisfies the following property for $i=1,2$, and Krogdahl \cite{krogdahl1974combinatorial,krogdahl1976combinatorial,krogdahl1977dependence} made its partial converse.

\begin{lemma}\label{lem:UPM-inv}
	Let $I\in \cI_i$ and let $Z\subseteq E$ satisfy $|I\triangle Z|=|I|$ and $I\triangle Z\in \cI_i$. 
	Then $A_i[I]$ contains a perfect matching on $Z$
	(i.e., a set of vertex-disjoint arcs such that $Z$ is the set of tails and heads of these arcs).
\end{lemma}

\begin{lemma}[Unique Perfect Matching Lemma]\label{lem:UPM}
	Let $I\in \cI_i$ and let $Z\subseteq E$ satisfy $|I\triangle Z|=|I|$.
	If $A_i[I]$ contains a unique perfect matching on $Z$, then $I\triangle Z\in \cI_i$.
\end{lemma}

Let us recall that a standard algorithm for finding a maximum-cardinality common independent set is driven by the following subroutine, Algorithm~\ref{alg:1} (see \cite[Section~41.2]{schrijver2003combinatorial}).

For any digraph $D=(E,A)$, a \textbf{path} in $D$ is a sequence $P=e_1e_2\cdots e_\ell$ of distinct vertices such that $(e_i,e_{i+1})\in A$ for each $i=1,2,\dots,\ell-1$; we call $P$ an \textbf{$e_1$--$e_\ell$ path} or \textbf{an $X$--$Y$ path} for sets $X \ni e_1$ and $Y \ni e_\ell$ to emphasize the end vertices, and define $\ell$ as the \textbf{length}.
A \textbf{cycle} in $D$ is a path that satisfies $(e_\ell, e_1)\in A$.
We often identify a path or cycle $e_1e_2\cdots e_\ell$ with its vertex set $\{e_1, e_2,\dots,e_\ell\}$.

\begin{algorithm2e}[ht!]
\caption{{{\sc Augment}$[E, \cI_1, \cI_2, I]$}} \label{alg:1}
\SetAlgoLined

\SetKwInOut{Input}{Input}\SetKwInOut{Output}{Output}
\Input{A finite set $E$, oracle access to the independence set families $\cI_1$ and $\cI_2$, and a common independent set $I \in \cI_1 \cap \cI_2$.}
\Output{A common independent set $J \in \cI_1 \cap \cI_2$ with $|J| = |I| + 1$ if one exists, or a subset $Z \subseteq E$ with $r_1(Z) + r_2(E \setminus Z) = |I|$.}
\BlankLine

Construct the exchangeability graph $D[I]$ with source set $S_I$ and sink set $T_I$.

If some $t \in T_I$ is reachable from some $s \in S_I$, then find a shortest $S_I$--$T_I$ path $P$ in $D[I]$, and return $J = I \triangle P$.

Otherwise, return $Z = \{\, e \in E \mid e~\text{can reach some}~t \in T_I~\text{in}~D[I] \,\}$.
\end{algorithm2e}

Now we turn to the weighted setting. For a weight function $w \colon E \to \RR$ and a subset $X \subseteq E$, define $w(X) \coloneqq \sum_{e \in X} w(e)$.
For a family $\cF \subseteq 2^E$, a subset $X \subseteq E$ is \textbf{$w$-maximal in $\cF$} if $X \in \textrm{arg\,max}\left\{\, w(Y) \mid Y \in \cF \,\right\}$. %
We define $\cI_i^k \coloneqq \{\, X \in \cI_i \mid |X| = k \,\}$ for $i = 1, 2$ and $k = 0, 1, \dots, n$.

One approach to solve the weighted matroid intersection problem is via augmentation along cheapest paths in the exchangeability graph  as shown in Algorithm~\ref{alg:2} (see~\cite[Section~41.3]{schrijver2003combinatorial}), where the cost function $c \colon E \to \RR$ is defined on the vertex set as follows:
  \begin{align}\label{eq:cost}
    c(e) &\coloneqq \begin{cases}
      w(e) & (e \in I),\\
      -w(e) & (e \in E \setminus I).
    \end{cases}
  \end{align}
For each path (or cycle) $P$ in $D[I]$, we define the \textbf{cost} of $P$ as $c(P) \coloneqq \sum_{e \in P}c(e)$.

\begin{algorithm2e}[ht!]
	\caption{{{\sc CheapestPathAugment}$[E, w, \cI_1, \cI_2, I]$}} \label{alg:2}
	\SetAlgoLined	
	\SetKwInOut{Input}{Input}\SetKwInOut{Output}{Output}
\Input{A finite set $E$, a weight function $w \colon E \to \RR$, oracle access to $\cI_1$ and $\cI_2$, and a $w$-maximal set $I \in \cI_1^k \cap \cI_2^k$ for some $k\in\{0, 1, \dots, n - 1\}$.}
\Output{A $w$-maximal set $J \in \cI_1^{k+1} \cap \cI_2^{k+1}$ if one exists, or a subset $Z \subseteq E$ with $r_1(Z) + r_2(E \setminus Z) = |I|$.}
\BlankLine

Create the exchangeability graph $D[I]$ with $S_I$ and $T_I$ as with {\sc Augment}$[E, \cI_1, \cI_2, I]$. In addition, define a cost function $c \colon E \to \RR$ by \eqref{eq:cost}.

If some $t \in T_I$ is reachable from some $s \in S_I$, then find a shortest cheapest $S_I$--$T_I$ path $P$ in $D[I]$ (i.e., the cost $c(P)$ is minimum, and subject to this, the length of $P$ is minimum), and return $J = I \triangle P$.

Otherwise, return $Z = \{\, e \in E \mid e~\text{can reach some}~t \in T_I~\text{in}~D[I] \,\}$.
\end{algorithm2e}

The next lemma characterizes $w$-maximal common independent sets in $\cI^k_1\cap \cI^k_2$.

\begin{lemma}[cf.~{\cite[Theorem~41.5]{schrijver2003combinatorial}}]\label{lem:negative_cycle}
A common independent set $I \in \cI_1^k \cap \cI_2^k$ is $w$-maximal if and only if $D[I]$ contains no negative-cost cycle with respect to the cost function $c$ defined as \eqref{eq:cost}.
\end{lemma}

We also remark that any even prefix or suffix of any shortest $S_I$--$T_I$ path in $D[I]$ is also exchangeable (as it has no shortcut, this is an easy consequence of Lemma~\ref{lem:UPM}).
It is similarly (e.g., by considering the weight-splitting algorithm~\cite{frank1981weighted}) observed that this is also true for any shortest cheapest $S_I$--$T_I$ path in $D[I]$ if $I$ is $w$-maximal in $\cI_1^k \cap \cI_2^k$, where $k = |I|$.

\begin{lemma}[cf.~{\cite[Lemmas~13.1.11 and 13.2.14]{frank2011connections}}]\label{lem:even_prefix}
For any $I \in \cI_1^k \cap \cI_2^k$, the following statements hold.
\begin{itemize}
\setlength{\itemsep}{0mm}
\item Let $P$ be a shortest $S_I$--$T_I$ path in $D[I]$, and $P'$ be a prefix of $P$ ending in $I$ or a suffix starting in $I$, which is of even length.
Then, $I \triangle P' \in \cI_1^k \cap \cI_2^k$. %
\item Suppose that $I$ is $w$-maximal in $\cI_1^k \cap \cI_2^k$.
Let $P$ be a shortest cheapest $S_I$--$T_I$ path in $D[I]$ with respect to the cost function $c$ defined as \eqref{eq:cost}, and $P'$ be a prefix of $P$ ending in $I$ or a suffix starting in $I$, which is of even length.
Then, $I \triangle P' \in \cI_1^k \cap \cI_2^k$. %
\end{itemize}
\end{lemma}

\section{Cardinality Matroid Intersection}\label{sec:unweighted_min_rank}
From now on, we assume access only to the minimum rank function $\rmin \colon 2^E \to \ZZ_{\geq 0}$ defined by
\begin{align*}
  \rmin(X) \coloneqq \min\left\{r_1(X), r_2(X)\right\} \quad (X \subseteq E).
\end{align*}
Note that $I \subseteq E$ is a common independent set if and only if $\rmin(I) = |I|$.

First, we observe that Theorem~\ref{thm:Edmonds} can be rephrased as follows.

\begin{lemma}\label{lem:duality_rmin}
The following equation holds:
\begin{align*}
  \max\left\{\, |I| \bigm| I \in \cI_1 \cap \cI_2 \,\right\} &= \min\left\{\, \rmin(Z) + \rmin(E \setminus Z) \mid Z \subseteq E \,\right\}.
\end{align*}
\end{lemma}

\begin{proof}
For any common independent set $I \in \cI_1 \cap \cI_2$ and any subset $Z \subseteq E$, we have
\begin{align*}
  |I| &= |I \cap Z| + |I \setminus Z| \leq \rmin(Z) + \rmin(E \setminus Z) \leq r_1(Z) + r_2(E \setminus Z).
\end{align*}
Since Theorem~\ref{thm:Edmonds} assures the minimum of the most right-hand side is equal to the maximum of the most left-hand side, the same is true for the intermediate one with respect to $\rmin$.
\end{proof}

From the algorithmic viewpoint, while Lemma~\ref{lem:duality_rmin} gives an optimality certification in cardinality maximization, it is difficult to construct the exchangeability graph $D[I]$ (cf.~Definition~\ref{def:exchange}) used in Algorithm~\ref{alg:1}.
Specifically, in general, we cannot exactly determine the fundamental circuits $C_i(I, x)$ for all $x \in E \setminus (I \cup S_I \cup T_I)$ and $i = 1, 2$ only by the minimum rank function.

We overcome this difficulty by modifying the exchangeability graph as follows.
Roughly speaking, we estimate the fundamental circuits $C_1(I, x)$ and $C_2(I, x)$ as $C_1(I, x) \cup C_1(I, t^*)$ and $C_2(I, x) \cup C_2(I, s^*)$, respectively, for some fixed $s^* \in S_I \setminus T_I$ and $t^* \in T_I \setminus S_I$.
This yields extra arcs not existing in the original exchangeability graph $D[I]$, but those extra arcs are not used in any shortest $S_I$--$T_I$ path in the new graph (cf.~Lemma~\ref{lem:AugmentMinRank}).
That is, the shortest $S_I$--$T_I$ paths are completely preserved, and hence one can find a desired shortest path using the modified exchangeability graph instead of $D[I]$.

\begin{definition}[Modified Exchangeability Graph]
For a common independent set $I \subseteq E$ and a pair $(s^*, t^*)$ with $\rmin(I + s^*) = \rmin(I + t^*) = \rmin(I) = |I|$ and $\rmin(I + s^* + t^*) = |I| + 1$, the \textbf{modified exchangeability graph} is a directed bipartite graph $\Dmin[I; s^*, t^*] = (E \setminus I, I; \Amin[I; s^*, t^*])$ defined as follows.
Define two subsets $S_I^*, T_I^* \subseteq E \setminus I$ by
\begin{align*}
  S_I^* &\coloneqq \{\, s \in E \setminus I \mid \rmin(I + s + t^*) = |I| + 1 \,\},\\
  T_I^* &\coloneqq \{\, t \in E \setminus I \mid \rmin(I + s^* + t) = |I| + 1 \,\}.
\end{align*}
We then define $\Amin[I; s^*, t^*] \coloneqq \Amino[I; s^*, t^*] \cup \Amint[I; s^*, t^*]$, where
\begin{align}
  \Amino[I; s^*, t^*] &\coloneqq \{\, (y, s) \mid s \in S_I^*,~y \in I \,\}\nonumber\\
  &\quad \cup \{\, (y, t) \mid t \in T_I^* \setminus S_I^*,~y \in I,~\rmin(I + t - y) = |I| \,\}\nonumber\\
  &\quad \cup \{\, (y, x) \mid x \in E \setminus (I \cup S_I^* \cup T_I^*),~y \in I,\ \rmin(I + t^* + x - y) = |I| \,\},\label{eq:Amino}\\
  \Amint[I; s^*, t^*] &\coloneqq \{\, (t, y) \mid t \in T_I^*,~y \in I \,\}\nonumber\\
  &\quad \cup \{\, (s, y) \mid s \in S_I^* \setminus T_I^*,~y \in I,~\rmin(I + s - y) = |I| \,\}\nonumber\\
  &\quad \cup \{\, (x, y) \mid x \in E \setminus (I \cup S_I^* \cup T_I^*),~y \in I,\ \rmin(I + s^* + x - y) = |I| \,\}.\label{eq:Amint}
\end{align}
\end{definition}

It is easy to observe that $\{S_I^*, T_I^*\} = \{S_I, T_I\}$, and we can assume that $S_I^* = S_I$ and $T_I^* = T_I$ by symmetry (i.e., by exchanging $s^*$ and $t^*$ if necessary).

\begin{lemma}\label{lem:AugmentMinRank}
The modified exchangeability graph $\Dmin[I; s^*, t^*]$ satisfies the following properties.
\begin{enumerate}[\rm (a)]
  \item $A_i[I] \subseteq \Amini[I; s^*, t^*]$ for $i = 1, 2$.
    Moreover, if $(u, v) \in \Amin[I; s^*, t^*] \setminus A[I]$, then $u, v \not\in S_I \cup T_I$.
  \item If $(y, x) \in \Amino[I; s^*, t^*] \setminus A_1[I]$, then $(y, t^*) \in A_1[I]$.
    If $(x, y) \in \Amint[I; s^*, t^*] \setminus A_2[I]$, then $(s^*, y) \in A_2[I]$.
\end{enumerate}
\end{lemma}

\begin{proof}
(a)~
By symmetry, we only prove the statement for $i = 2$.
We first confirm that $A_2[I]$ and $\Amint[I]$ coincide around each vertex in $S_I \cup T_I$.
This is trivial for each sink in $T_I = T_I^*$ by definition.
Let $s \in S_I = S_I^*$ and $y \in I$.
Then, $(s, y) \in A_2[I]$ if and only if $I + s - y \in \cI_2$ by definition.
As $I + s \in \cI_1$ by $s \in S_I$, we have $r_1(I + s - y) = r_1(I + s) - 1 = |I|$.
Hence, $(s, y) \in \Amint[I; s^*, t^*]$ if and only if $r_2(I + s - y) \geq |I|$, which means $I + s - y \in \cI_2$.
Thus, $A_2[I]$ and $\Amint[I]$ coincide around each source $s \in S_I$, and we are done.

Next, suppose that $(x, y) \in A_2[I]$ with $x \in E \setminus (I \cup S_I \cup T_I)$ and $y \in I$, and we show that $(x, y) \in \Amint[I; s^*, t^*]$, which completes the proof.
By definition, we have $I + x - y \in \cI_2$, and we show $\rmin(I + s^* + x - y) = |I|$, regardless of the choice of $s^* \in S_I \setminus T_I$.
On one hand, as $I + s^* \in \cI_1$ (by $s^* \in S_I$), we have $r_1(I + s^* + x - y) \geq r_1(I + s^*) - 1 = |I|$.
On the other hand, as $\{s^*, x\} \subseteq \cl_2(I)$ by $s^*, x \not\in T_I$, we have $r_2(I + s^* + x) = |I|$, and hence $|I| = r_2(I + x - y) \leq r_2(I + s^* + x - y) \leq r_2(I + s^* + x) = |I|$.
This concludes that $\rmin(I + s^* + x - y) = |I|$, i.e., $(x, y) \in \Amint[I; s^*, t^*]$.
Thus we are done.

\medskip\noindent
(b)~
By symmetry, fix an arc $(x, y) \in \Amint[I; s^*, t^*] \setminus A_2[I]$ with $x \in E \setminus (I \cup S_I \cup T_I)$ and $y \in I$ (cf.~(a)), and we show $r_2(I + s^* - y) = |I|$, which implies $(s^*, y) \in A_2[I]$.
By definition, we have $\rmin(I + s^* + x - y) = |I|$ and $r_2(I + x - y) \leq |I| - 1$. As $x \neq s^* \not\in T_I$, we obtain
\begin{align*}
  |I| &= r_2(I + s^*) \geq r_2(I + s^* - y) \geq r_2(I + s^* + x - y) + r_2(I - y) - r_2(I + x - y) \geq |I|,
\end{align*}
where the middle inequality is by the submodularity of the rank function $r_2$.
This concludes that $r_2(I + s^* - y) = |I|$, and we are done.
\end{proof}

By Lemma~\ref{lem:AugmentMinRank}, for any pair $(s^*, t^*)$ with $s^* \in S_I \setminus T_I$ and $t^* \in T_I \setminus S_I$,
\begin{itemize}
\item the modified exchangeability graph $\Dmin[I; s^*, t^*]$ includes the exchangeability graph $D[I]$ as its subgraph, and
\item any shortest $S_I$--$T_I$ path in $\Dmin[I; s^*, t^*]$ consists of arcs in $A[I]$, since there must be a shortcut arc in $A[I]$ if some arc in $\Amin[I; s^*, t^*] \setminus A[I]$ is traversed.
\end{itemize}
Hence, Algorithm~\ref{alg:augment_min_rank} correctly emulates the usual augmentation step (Algorithm~\ref{alg:1}) up to the symmetry of $\bM_1$ and $\bM_2$.
Note that if there exists $x \in S_I \cap T_I$, then we can recognize it as $\rmin(I + x) = |I| + 1$ and augment $I$ just by adding $x$ to $I$.
We also remark that we do not need to distinguish the two cases when $s^* \in S_I$ (and $t^* \in T_I$) and when $s^* \in T_I$ (and $t^* \in S_I$) in Algorithm~\ref{alg:augment_min_rank}; the latter case corresponds to the case when we just interchange $\bM_1$ and $\bM_2$ in advance (by the definitions of $\Amino[I; s^*, t^*]$ and $\Amint[I; s^*, t^*]$ as well as $S_I^*$ and $T_I^*$), which does not essentially change the task.
Since the construction of $\Dmin[I; s^*, t^*]$ requires $\mathrm{O}(n^2)$ time (including the number of oracle accesses) and the number of augmentations is the maximum size of a common independent set, we obtain the following theorem.

\begin{theorem}\label{thm:unweighted}
   A maximum common independent set in the two matroids can be found in $\mathrm{O}(rn^2)$ time only using the minimum rank oracle, where $r$ is the maximum cardinality of a common independent set.
\end{theorem}

\begin{algorithm2e}[ht!]
\caption{{{\sc AugmentMinRank}$[E, \rmin, I]$}} \label{alg:augment_min_rank}
\SetAlgoLined

\SetKwInOut{Input}{Input}\SetKwInOut{Output}{Output}
\Input{A finite set $E$, oracle access to the minimum rank function $\rmin$ of two matroids on $E$, and a common independent set $I$ (with $\rmin(I) = |I|$).}
\Output{A common independent set $J \subseteq E$ with $\rmin(J) = |J| = |I| + 1$ if one exists, or a subset $Z \subseteq E$ with $\rmin(Z) + \rmin(E \setminus Z) = |I|$.}
\BlankLine

If $\rmin(I + s + t) = |I|$ for every $s, t \in E \setminus I$, then just halt.%

If $\rmin(I + x) = |I| + 1$ for some $x \in E \setminus I$, then halt with returning $J = I + x$.

Take any pair $(s^*, t^*)$ with $\rmin(I + s^*) = \rmin(I + t^*) = |I|$ and $\rmin(I + s^* + t^*) = |I| + 1$,
and create the modified exchangeability graph $\Dmin[I; s^*, t^*]$ with $S_I^*$ and $T_I^*$.

If some $s \in S_I^*$ can reach some $t \in T_I^*$, then find a shortest $S_I^*$--$T_I^*$ path $P$ in $\Dmin[I; s^*, t^*]$, and return $J = I \triangle P$.
Otherwise, return $Z = \{\, e \in E \mid e~\text{can reach}~T_I^*~\text{in}~\Dmin[I; s^*, t^*] \,\}$.
\end{algorithm2e}

\section{Consistent Exchangeability Graphs}\label{sec:consistent}
For the weighted problem, we try to find a shortest cheapest $S_I$--$T_I$ path in $D[I]$ using the modified exchangeability graph.
In contrast to cardinality augmentation, the extra arcs may do harm in the following two senses: %
\begin{itemize}
\item they yield extra cycles whose cost may be negative, thus making it difficult just to find a cheapest path in the graph;
\item some extra arcs may be traversed by a shortest cheapest $S_I^*$--$T_I^*$ path in $\Dmin[I; s^*, t^*]$, for which we may not obtain a better path in $D[I]$ by using a shortcut arc $(s^*, y)$ or $(y, t^*)$ (cf.~Lemma~\ref{lem:AugmentMinRank}).
\end{itemize}
In order to overcome these issues, we refine the modified exchangeability graph in two steps.

Let us fix a common independent set $I$ (whose $w$-maximality is not necessary in most parts, and we explicitly mention when we need it).
First, we define $\Dmin[I] = (E \setminus I, I; \Amin[I])$ as the intersection of all possible candidates for the modified exchangeability graphs after $S_I^*$ and $T_I^*$ are defined by fixing $s^* \in S_I^* \setminus T_I^*$ and $t^* \in T_I^* \setminus S_I^*$ with $\rmin(I + s^*) = \rmin(I + t^*) = |I|$ and $\rmin(I + s^* + t^*) = |I| + 1$.
In other words, we define $\Amin[I] \coloneqq \Amino[I] \cup \Amint[I]$, where %
\begin{align*}
\Amino[I] &\coloneqq \bigcap_{t \in T_I^* \setminus S_I^*} \Amino[I; s^*, t],\\
\Amint[I] &\coloneqq \bigcap_{s \in S_I^* \setminus T_I^*} \Amint[I; s, t^*].
\end{align*}
We here remark that $\Amino[I; s, t]$ and $\Amint[I; s, t]$ depend only on $t \in T_I^* \setminus S_I^*$ and $s \in S_I^* \setminus T_I^*$, respectively (cf.~\eqref{eq:Amino} and \eqref{eq:Amint}).
As with in the previous section, we assume that $S_I^* = S_I$ and $T_I^* = T_I$, and then Lemma~\ref{lem:AugmentMinRank} is strengthened as follows.

\begin{lemma}\label{lem:intersection}
  The graph $\Dmin[I]$ satisfies the following properties.
\begin{enumerate}[\rm (a)]
  \item[{\rm (a)}] $A_i[I] \subseteq \Amini[I]$ for $i = 1, 2$.
    Moreover, if $(u, v) \in \Amin[I] \setminus A[I]$, then $u, v \not\in S_I \cup T_I$.
  \item[{\rm (b)}] If $(y, x) \in \Amino[I] \setminus A_1[I]$, then $(y, t) \in A_1[I]$ for every $t \in T_I$.
    If $(x, y) \in \Amint[I] \setminus A_2[I]$, then $(s, y) \in A_2[I]$ for every $s \in S_I$.
\end{enumerate}
\end{lemma}

In what follows, we call the exchangeability graph $D[I] = (E \setminus I, I; A[I])$ the \textbf{true graph}.
An arc $a \in \Amin[I]$ is \textbf{true} if $a \in A[I]$, and \textbf{fake} otherwise.
By Lemma~\ref{lem:intersection}, we can immediately determine that an arc $a \in \Amin[I]$ is true if
\begin{itemize}
\item $a$ is incident to $S_I \cup T_I$,
\item $a = (y, x) \in \Amino[I]$ and $(y, t) \not\in \Amino[I]$ for some $t \in T_I$, or
\item $a = (x, y) \in \Amint[I]$ and $(s, y) \not\in \Amint[I]$ for some $s \in S_I$.
\end{itemize}
These arcs are called \textbf{sure}, and the other arcs in $\Amin[I]$ are called \textbf{suspicious}.
We denote by $\Asure[I] \subseteq \Amin[I]$ the set of sure arcs, and define $\Asusp[I] \coloneqq \Amin[I] \setminus \Asure[I]$.
In other words, we have $\Asusp[I] = \Asuspo[I] \cup \Asuspt[I]$, where
\begin{align*}
  \Asuspo[I] &\coloneqq \left\{\, (y, x) \in \Amino[I] \bigm| x \not\in S_I \cup T_I,~(y, t) \in \Amino[I]~(\forall t \in T_I \setminus S_I) \,\right\},\\
  \Asuspt[I] &\coloneqq \left\{\, (x, y) \in \Amint[I] \bigm| x \not\in S_I \cup T_I,~(s, y) \in \Amint[I]~(\forall s \in S_I \setminus T_I) \,\right\}.
\end{align*}

If $\Asusp[I] = \emptyset$, then we conclude that $\Dmin[I]$ is indeed the true graph $D[I]$, and we have to do nothing more.
Otherwise, as the second step, we further refine $\Dmin[I]$ by fixing a possible true-fake configuration of the suspicious arcs.
To make a reasonable refinement, we check small local exchanges.
Specifically, we observe the values $\rmin((I \cup X) \setminus Y)$ for all pairs $(X, Y)$ of $X \subseteq E \setminus (I \cup S_I \cup T_I)$ and $Y \subseteq I$ with $1 \leq |X| \leq 2$ and $1 \leq |Y| \leq 2$.
We define the consistency with those observations as follows; roughly speaking, a local exchange decreases the minimum rank value not so much (by less than the number of removed elements) if and only if the local exchange contains an exchangeable pair in each matroid.

\begin{definition}\label{def:LE-pair}
A pair $(X, Y)$ of $X \subseteq E \setminus (I \cup S_I \cup T_I)$ and $Y \subseteq I$ with $1 \leq |X| \leq 2$ and $1 \leq |Y| \leq 2$ is called a \textbf{local exchange pair} or an \textbf{LE-pair} for short.
An LE-pair $(X', Y')$ is called a \textbf{subpair} of an LE-pair $(X, Y)$ if $X' \subseteq X$ and $Y' \subseteq Y$, and a \textbf{proper subpair} if at least one inclusion is strict.
\end{definition}

\begin{definition}\label{def:consistent}
We say that a directed bipartite graph $\tD[I] = (E \setminus I, I; \tA[I])$ is \textbf{consistent} with respect to an LE-pair $(X, Y)$ if the following holds:
\begin{itemize}
    \item If $\rmin((I \cup X) \setminus Y) \ge |I| - |Y| + 1$, then $\tA[I] \cap (Y \times X) \neq \emptyset$ and $\tA[I] \cap (X \times Y) \neq \emptyset$.
    \item Otherwise (if $\rmin((I \cup X) \setminus Y) = |I| - |Y|$), $\tA[I] \cap (Y \times X) = \emptyset$ or $\tA[I] \cap (X \times Y) = \emptyset$.
\end{itemize}

We also say that $\tD[I]$ is \textbf{overestimated} or \textbf{underestimated} with respect to an LE-pair $(X, Y)$ if $\tD[I]$ can be made consistent with respect to $(X, Y)$ by removing or adding (possibly zero) arcs, respectively.
In particular, $\tD[I]$ is consistent with respect to $(X, Y)$ if and only if $\tD[I]$ is both overestimated and underestimated with respect to $(X, Y)$.
\end{definition}

We provide some examples of LE-pairs and estimation using those pairs.

\begin{example}\label{ex:LE-pair1}
Consider an LE-pair $(X, Y)$ with $X = \{x_1, x_2\}\subseteq E\setminus (I\cup S_I\cup T_I)$ and $Y = \{y\}\subseteq I$. Suppose that we have $\rmin(I+x_1+x_2-y)=|I|$, $\rmin(I+x_1-y)=|I|-1$, and $\rmin(I+x_2-y)=|I|-1$. Then, a graph $\tD[I]$ is consistent with respect to all the subpairs of $(X, Y)$ if and only if the restriction of $\tD[I]$ to $X\cup Y$ has an arc set $\{(x_1, y), (y,x_2)\}$ or $\{(x_2,y), (y,x_1)\}$. See (A) and (B) in Figure~\ref{fig:LE-pair1}. (C) and (D) in Figure~\ref{fig:LE-pair1} are examples of overestimation and underestimation, respectively.
\end{example}

\begin{example}\label{ex:LE-pair2}
Consider an LE-pair $(X, Y)$ with $X = \{x_1, x_2\}\subseteq E\setminus (I\cup S_I\cup T_I)$ and $Y = \{y_1, y_2\}\subseteq I$. Suppose that we have $\rmin(I+x_1+x_2-y_1-y_2)=|I|-1$, $\rmin(I+x_1+x_2-y_i)=|I|-1~(i\in \{1,2\})$, $\rmin(I+x_i-y_1-y_2)=|I|-2~(i\in \{1,2\})$, and $\rmin(I+x_i-y_j)=|I|-1~(i,j\in \{1,2\})$. Then, a graph $\tD[I]$ is consistent with respect to all the subpairs of $(X, Y)$ if and only if the restriction of $\tD[I]$ to $X\cup Y$ has an arc set either $\{(x_1, y_1), (y_2,x_2)\}$, $\{(x_2, y_2), (y_1, x_1)\}$, $\{(x_1, y_2), (y_1, x_2)\}$, or $\{(x_2, y_1), (y_2, x_1)\}$.
Figure~\ref{fig:LE-pair2} shows these four possible arc sets. We will call this kind of LE-pairs {\bf evil} (Definition~\ref{def:evil_pair}) as they will be the main obstacle in our approach described in Section~\ref{sec:2-SAT}.
\end{example}

\begin{figure}[t]
\begin{center}
\includegraphics[width=\textwidth]{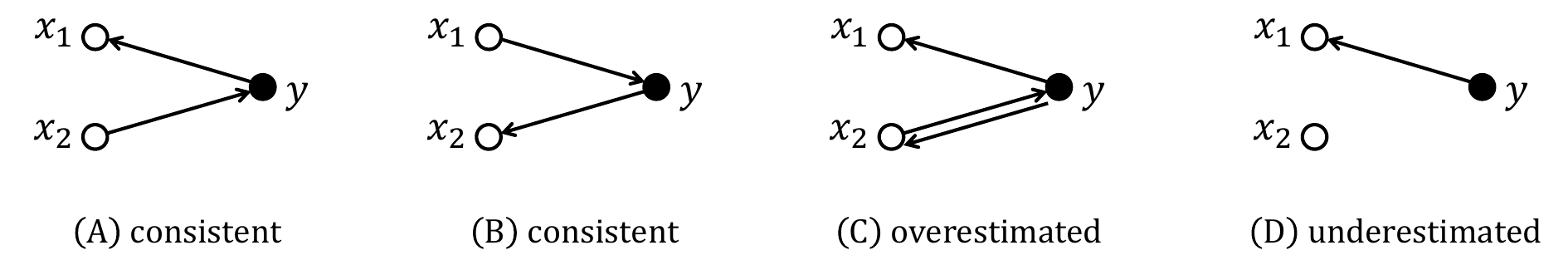}
\caption{An illustration for Example \ref{ex:LE-pair1}.}
\label{fig:LE-pair1}
\end{center}
\end{figure}

\begin{figure}[t]
\begin{center}
\includegraphics[width=\textwidth]{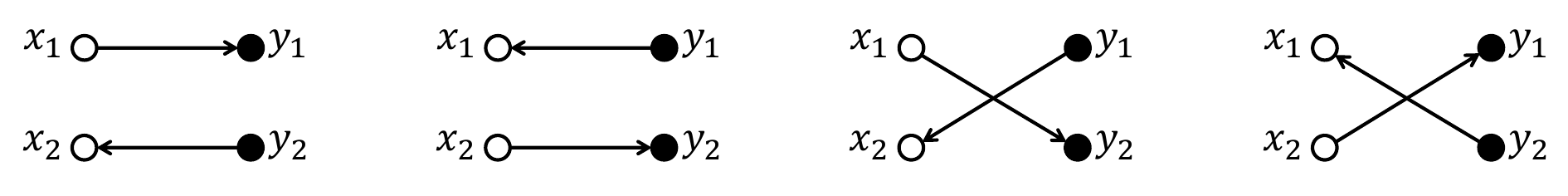}
\caption{An illustration for Example \ref{ex:LE-pair2}. All the four situations are consistent.}
\label{fig:LE-pair2}
\end{center}
\end{figure}

\begin{definition}\label{def:consistent_graph}
A directed bipartite graph $\tD[I] = (E \setminus I, I; \tA[I])$ is called a \textbf{consistent exchangeability graph} (or, simply, \textbf{consistent}) if
\begin{itemize}
    \item $\Asure[I] \subseteq \tA[I] \subseteq \Amin[I]$, and
    \item $\tD[I]$ is consistent with respect to all LE-pairs $(X, Y)$.
\end{itemize}

We also say that $\tD[I]$ is an \textbf{over}/\textbf{underestimation} (or \textbf{over}/\textbf{underestimated}) if the above two conditions with replacing ``consistent'' respectively by ``over/underestimated'' hold. 
\end{definition}

Note that a subgraph of a consistent exchangeability graph including all the sure arcs is an underestimation, but the converse (an underestimation is such a subgraph of some consistent exchangeability graph) is not necessarily true.
The condition of underestimation just requires that it is underestimated with respect to each LE-pair.

It is easy to observe from the definition that the true graph $D[I]$ is always consistent.
The following lemma guarantees that any consistent exchangeability graph $\tD[I]$ is suitable for our purpose of finding a shortest cheapest $S_I$--$T_I$ path in $D[I]$.
Specifically, the property (2) implies that $D[I]$ and $\tD[I]$ have the same shortest cheapest $S_I$--$T_I$ paths in terms of the vertex sets.
We define $\cImink \coloneqq \{\, I \subseteq E \mid \rmin(I) = |I| = k \,\} = \cI_1^k \cap \cI_2^k$ for each $k = 0, 1, \dots, n$.

\begin{lemma}\label{lem:correctness}
Let $\tD[I]$ be a consistent exchangeability graph, and $c$ be the cost function defined as \eqref{eq:cost}.
If $I$ is $w$-maximal in $\mathcal{I}_\mathrm{min}^{|I|}$, then the following properties hold.
\begin{enumerate}[\rm (1)]
\item $\tD[I]$ has no negative-cost cycle.
\item For any shortest cheapest $S_I$--$T_I$ path $\tP$ in $\tD[I]$, there exists an $S_I$--$T_I$ path $P$ in $D[I]$ with $P = \tP$ as the vertex sets.
  The same is true if we interchange $\tD[I]$ and $D[I]$.
\end{enumerate}
\end{lemma}

Before proving this lemma, we prepare two specific lemmas. %

\begin{lemma}\label{lem:cycle_partition}
Let $\tD[I]$ and $\tD'[I]$ be a consistent and underestimated exchangeability graph, respectively.
Then, any cycle $\tQ'$ in $\tD'[I]$ can be partitioned into disjoint cycles in $\tD[I]$.
\end{lemma}

\begin{proof}
Fix a cycle $\tQ'$ in $\tD'[I]$.
It suffices to show that there exist perfect matchings on $\tQ'$ both in $\tA_1[I] = \tA[I] \cap (I \times (E \setminus I))$ and in $\tA_2[I] = \tA[I] \cap ((E \setminus I) \times I)$, whose union partitions $\tQ'$ into cycles in $\tD[I] = (E \setminus I, I; \tA[I])$.
By symmetry, we focus on $\tA_2[I]$.

Suppose to the contrary that $\tA_2[I]$ has no perfect matching on $\tQ'$.
Then, by Hall's theorem~\cite{hall1935representatives} (see, e.g.,~\cite[Theorem~22.1]{schrijver2003combinatorial}), there exists a subset $Y \subseteq \tQ' \cap I$ such that $|X| < |Y|$, where $X \coloneqq \Gamma_{\tA_2[I]}(Y) \cap (\tQ' \setminus I)$ and $\Gamma_{\tA_2[I]}(Y) \coloneqq \{\, x \mid (x, y) \in \tA_2[I],~y \in Y \,\}$.
Since $\tQ'$ is a cycle in $\tD'[I]$, it contains two arcs $(z_1, y_1), (y_2, z_2)$ such that $z_1, z_2 \in \tQ' \setminus (I \cup X)$ and $y_1, y_2 \in Y$ (possibly, $z_1 = z_2$ or $y_1 = y_2)$.
Let $Z' = \{z_1, z_2\}$ and $Y' = \{y_1, y_2\}$.
Since $\tD'[I]$ is underestimated with respect to $(Z', Y')$, we have $\rmin((I \cup Z') \setminus Y') \geq |I| - |Y'| + 1$.
Then, the consistency of $\tD[I]$ with respect to $(Z', Y')$ implies that $(z_i, y_j) \in \tA_2[I]$ for some $i, j \in \{1, 2\}$, which contradicts the definition of $X$.
\end{proof}

\begin{lemma}\label{lem:path_cycle_partition}
Let $\tD[I]$ and $\tD'[I]$ be a consistent and underestimated exchangeability graph, respectively.
Then, any $S_I$--$T_I$ path $\tP'$ in $\tD'[I]$ can be partitioned into a disjoint family of an $S_I$--$T_I$ path and cycles in $\tD[I]$.
\end{lemma}

\begin{proof}
Fix an $S_I$--$T_I$ path $\tP'$ in $\tD'[I]$, and let $s \in S_I$ and $t \in T_I$ denote its end vertices.
It suffices to show that there exist perfect matchings on $\tP' - s$ in $\tA_1[I]$ and on $\tP' - t$ in $\tA_2[I]$, whose union forms a desired family.
By symmetry, we focus on $\tA_2[I]$.

Suppose to the contrary that $\tA_2[I]$ has no perfect matching on $\tP' - t$.
Then, by Hall's theorem, again, there exists a subset $Y \subseteq \tP' \cap I$ such that $|X| < |Y|$, where $X \coloneqq \Gamma_{\tA_2[I]}(Y) \cap (\tP' \setminus (I + t))$.
Since $\tP'$ is an $s$--$t$ path in $\tD'[I]$, it contains an arc $a = (z_1, y_1)$ such that $z_1 \in \tP' \setminus (I \cup X + t)$ and $y_1 \in Y$.
By the definition of $X$, we have $a \not\in \tA_2[I]$, and hence $a \in \Asusp[I]$ (as $\tD'[I]$ is an underestimation, $\Asuret[I]\subseteq \tA_2[I]\subseteq \Asuret[I]\cup \Asuspt[I]$).
We then have $z_1 \neq s \in X$ by Lemma~\ref{lem:intersection}.
Since the $s$--$t$ path $\tP'$ does not use any arc entering $s \in X$ and we have $|X| < |Y|$, it uses at least two arcs from $Y$ to $\tP' \setminus (I \cup X)$.
Hence, $\tP'$ contains an arc $(y_2, z_2)$ such that $z_2 \in \tP' \setminus (I \cup X + t)$ and $y_2 \in Y$ (possibly, $z_1 = z_2$ or $y_1 = y_2)$.
Let $Z' = \{z_1, z_2\}$ and $Y' = \{y_1, y_2\}$.
Since $\tD'[I]$ is underestimated with respect to $(Z', Y')$, we have $\rmin((I \cup Z') \setminus Y') \geq |I| - |Y'| + 1$.
Then, the consistency of $\tD[I]$ implies that $(z_i, y_j) \in \tA_2[I]$ for some $i, j \in \{1, 2\}$, which contradicts the definition of $X$.
\end{proof}

Now we are ready to prove Lemma~\ref{lem:correctness}.

\begin{proof}[Proof of Lemma~\ref{lem:correctness}]
We obtain the property (1) as follows.
By Lemma~\ref{lem:cycle_partition}, any cycle $\tQ$ in $\tD[I]$ can be partitioned into disjoint cycles in $D[I]$, each of whose cost is nonnegative by Lemma~\ref{lem:negative_cycle}.

Similarly, we obtain the property (2) as follows.
Fix a shortest cheapest $S_I$--$T_I$ path $\tP$ in $\tD[I]$.
Then, by Lemma~\ref{lem:path_cycle_partition}, $\tP$ can be partitioned into a disjoint family of an $S_I$--$T_I$ path $P$ and cycles in $D[I]$.
By Lemma~\ref{lem:path_cycle_partition} (with interchanging the role of $\tD[I]$ and $\tD'[I]$), $P$ can be partitioned into a disjoint family of an $S_I$--$T_I$ path $\tilde{P}'$ and cycles in $\tD[I]$.
We then have $|\tP| \geq |P| \geq |\tilde{P}'|$.
Moreover, by Lemma~\ref{lem:negative_cycle} and the property (1), we have $c(\tP) \geq c(P) \geq c(\tilde{P}')$.
As $\tP$ is a shortest cheapest $S_I$--$T_I$ path in $\tD[I]$, all the inequalities hold with equality, which implies $P = \tP$.
The interchanged version of the property (2) is shown in the same manner, which completes the proof.
\end{proof}

Suppose that $I$ is $w$-maximal in $\mathcal{I}_\mathrm{min}^{|I|}$.
By Lemma~\ref{lem:correctness}, assuming that a consistent exchangeability graph can be found (in Line~\ref{line:step4}), Algorithm~\ref{alg:cheapest_path_augment_min_rank} correctly emulates the usual augmentation step (Algorithm~\ref{alg:2}) up to the symmetry of $\bM_1$ and $\bM_2$.
Note that when $\emptyset \neq S_I \subseteq T_I$ or $\emptyset \neq T_I \subseteq S_I$, some single element in $S_I \cap T_I$ forms a shortest cheapest $S_I$--$T_I$ path (as any even-length prefix and suffix of any $S_I$--$T_I$ path have nonnegative cost by Lemma~\ref{lem:negative_cycle}).

\begin{algorithm2e}[ht!]
\caption{{{\sc CheapestPathAugmentMinRank}$[E, \rmin, I]$}} \label{alg:cheapest_path_augment_min_rank}
\SetAlgoLined

\SetKwInOut{Input}{Input}\SetKwInOut{Output}{Output}
\Input{A finite set $E$, a weight function $w \colon E \to \RR$, oracle access to the minimum rank function $\rmin$, and a $w$-maximal set $I \in \cImink$ for some $k = 0, 1, \dots, n - 1$.}
\Output{A $w$-maximal set $J \in \cIminkk$ if one exists, or a subset $Z \subseteq E$ with $\rmin(Z) + \rmin(E \setminus Z) = k$.}
\BlankLine

If $\rmin(I + s + t) = |I|$ for every $s, t \in E \setminus I$, then just halt.%

If no pair $(s, t)$ satisfies $\rmin(I + s) = \rmin(I + t) = |I|$ and $\rmin(I + s + t) = |I| + 1$,
then return $J = I + x^*$ for any $x^* \in \arg\,\max \left\{\, w(x) \mid x \in E \setminus I,~\rmin(I + x) = |I| + 1 \,\right\}$.

Compute $S_I^*$ and $T_I^*$ by fixing a pair $(s^*, t^*)$ that satisfies $\rmin(I + s^*) = \rmin(I + t^*) = |I|$ and $\rmin(I + s^* + t^*) = |I| + 1$,
and create the graph $\Dmin[I]$ by taking the intersection of $\Dmin[I; s^*, t]$ and $\Dmin[I; s, t^*]$ for all $t \in T_I^* \setminus S_I^*$ and $s \in S_I^* \setminus T_I^*$, respectively.

Observe the values $\rmin((I \cup X) \setminus Y)$ for all LE-pairs $(X, Y$), and \textbf{find a consistent exchangeability graph} $\tD[I]$.\label{line:step4}

If some $s \in S_I^*$ can reach some $t \in T_I^*$, then find a shortest cheapest $S_I^*$--$T_I^*$ path $\tP$ in $\tD[I]$, and return $J = I \triangle \tP$.
Otherwise, return $Z = \{\, e \in E \mid e~\text{can reach}~T_I^*~\text{in}~\tD[I] \,\}$.
\end{algorithm2e}

\begin{lemma}\label{lem:correctness_consistent}
  Algorithm~\ref{alg:cheapest_path_augment_min_rank} returns a correct output.
\end{lemma}

The only remaining issue is how to find a consistent exchangeability graph efficiently. %
Unfortunately, this is NP-hard in general as follows (proved in Section~\ref{sec:NP-hard}), but we can find an ``almost consistent'' exchangeability graph, which is enough to solve several special cases by this approach (Section~\ref{sec:2-SAT}).

\begin{theorem}\label{thm:NP-hard}
The following problem is NP-hard even if the two matroids $\mathbf{M}_1$ and $\mathbf{M}_2$ are both $\mathbb{Q}$-representable:
Given a common independent set $I$, the source and sink sets $S_I$ and $T_I$, and the values $\rmin((I \cup X) \setminus Y)$ for all the LE-pairs $(X, Y)$, find a consistent exchangeability graph $\tD[I]$.
\end{theorem}

\begin{remark}
Lemma~\ref{lem:cycle_partition} together with Lemma~\ref{lem:negative_cycle} implies that the weighted matroid intersection problem under the minimum rank oracle enjoys an $(\mathrm{NP} \cap \mathrm{coNP})$-like feature.
On the one hand, if the given instance has a common independent set $I$ with $|I| = k$ and $w(I) \ge \alpha$, then $I$ itself is a certificate.
On the other hand, if the given instance has no such common independent set, then a pair of a $w$-maximal set $I$ in $\cImink$ (with $w(I) < \alpha$) and a consistent exchangeability graph $\tD[I]$ having no negative cycle is a certificate.\footnote{We remark that the LP duality also gives such a certificate. It is well-known that the dual LP of a usual description of the weighted matroid intersection problem (using the two rank functions) has a sparse optimal solution consisting of two chains (cf.~\cite[Theorem~41.12]{schrijver2003combinatorial}), and this is true for the description using the minimum rank function.}
This may be one positive evidence that enables us to hope that this problem is tractable in general.
\end{remark}

\section{Almost Consistent Exchangeability Graphs via 2-SAT}\label{sec:2-SAT}
In this section, we show that an \textbf{almost consistent} exchangeability graph can be obtained in polynomial time by solving the 2-SAT problem, which is enough in several special cases.
Specifically, we observe that all but one condition of the consistency in Definition~\ref{def:consistent} can be formulated by 2-CNF.
The following definition corresponds to the unique difficult case, which is also utilized for the proof of NP-hardness of finding a consistent exchangeability graph in general.

\begin{definition}[cf. Example~\ref{ex:LE-pair2}]\label{def:evil_pair}
An LE-pair $(X, Y)$ is said to be \textbf{evil} if
\begin{itemize}
    \item $|X| = |Y| = 2$,
    \item $\rmin((I \cup X) \setminus Y) = |I| - 1 \ (= |I| - |Y| + 1)$, and
    \item $\rmin((I \cup X') \setminus Y') = |I| - |Y'|$ for every proper subpair $(X', Y')$ of $(X, Y)$.
\end{itemize}
\end{definition}

\begin{definition}\label{def:almost_consistent_graph}
A directed bipartite graph $\tD[I] = (E \setminus I, I; \tA[I])$ is called an \textbf{almost consistent exchangeability graph} (or, simply, \textbf{almost consistent}) if
\begin{itemize}
    \item $\Asure[I] \subseteq \tA[I] \subseteq \Amin[I]$,
    \item $\tD[I]$ is consistent with respect to all non-evil LE-pairs $(X, Y)$,
    \item $\tD[I]$ is underestimated with respect to all evil LE-pairs $(X, Y)$, and
    \item if $\tD[I]$ is not consistent with respect to an evil LE-pair $(X, Y)$, then $\tA[I]$ contains no arc between $X$ and $Y$.
\end{itemize}
In particular, any consistent exchangeability graph $D[I]$ is almost consistent.
\end{definition}

\subsection{2-SAT Formulation}
To find an almost consistent exchangeability graph $\tD[I] = (E \setminus I, I; \tA[I])$, for each suspicious arc $a \in \Asusp[I]$, we introduce one associated variable $z_a$ such that $z_a = \true$ means $a \in \tA[I]$.
For the sake of simplicity, we introduce a constant $z_a = \true$ for each arc $a \in \Asure[I]$, and $z_a = \false$ for each pair $a = (u, v) \not\in \Amin[I]$ with either $u \in E \setminus (I \cup S_I \cup T_I)$ and $v \in I$ or $u \in I$ and $v \in E \setminus (I \cup S_I \cup T_I)$.

Let $r$ denote the maximum cardinality of a common independent set.
Then, the number of variables is $\mathrm{O}(rn)$, and the following lemma implies that one can obtain an almost consistent exchangeability graph $\tD[I] = (E \setminus I, I; \tA[I])$ in $\mathrm{O}(r^2n^2)$ time, since 2-SAT can be solved in time linear in the number of variables and clauses of the input 2-CNF~\cite{krom1967decision,aspvall1979linear}.

\begin{lemma}\label{lem:2-SAT}
    There exists a 2-CNF $\psi(z)$ on the variables $z = (z_a)_{a \in \Asusp[I]}$ consisting of $\mathrm{O}(r^2n^2)$ clauses such that $\psi(z)$ is satisfiable and each satisfying assignment $\tilde{z} \in \{\true, \false\}^{\Asusp[I]}$ gives an almost consistent exchangeability graph $\tD[I] = (E \setminus I, I; \tA[I])$ with $\tA[I] = \Asure[I] \cup \{\, a \in \Asusp[I] \mid \tilde{z}_a = \true \,\}$.
\end{lemma}

\begin{proof}
We define a set $\cC_{X, Y}$ of at most eight clauses, each of which contains at most two variables for each LE-pair $(X, Y)$. %
Note that there are $\mathrm{O}(r^2n^2)$ such pairs.

\subsubsection*{Case 1.1.~ When $|X| = 1$ and $|Y| = 1$.}
Suppose that $X = \{x\}$ and $Y = \{y\}$, and let $a = (x, y)$ and $b = (y, x)$.
We then have two possibilities $\rmin(I + x - y) = |I|$ and $\rmin(I + x - y) = |I| - 1$.

When $\rmin(I + x - y) = |I|$, both $a$ and $b$ must be included in $\tD[I]$. %
We then define $\cC_{X, Y} \coloneqq \{C_{X, Y}^1, C_{X, Y}^2, C_{X, Y}^3\}$ with
\[C_{X, Y}^1 \coloneqq (z_a \vee z_b), \quad C_{X, Y}^2 \coloneqq (\neg z_a \vee z_b), \quad C_{X, Y}^3 \coloneqq (z_a \vee \neg z_b).\]
It is easy to see that $\cC_{X, Y}$ is satisfied if and only if $z_a = z_b = \true$, which coincides with the desired consistency condition.

When $\rmin(I + x - y) = |I| - 1$, at least one of $a$ and $b$ must be excluded from $\tD[I]$. %
We then define $\cC_{X, Y} \coloneqq \{C_{X, Y}^1\}$ with
\[C_{X, Y}^1 \coloneqq (\neg z_a \vee \neg z_b).\]
It is easy to see that $\cC_{X, Y}$ is satisfied if and only if $z_a = \false$ or $z_b = \false$, which coincides with the desired consistency condition.

\subsubsection*{Case 1.2.~ When $|X| = 1$ and $|Y| = 2$.}
Suppose that $X = \{x\}$ and $Y = \{y_1, y_2\}$, and let $a_i = (x, y_i)$ and $b_i = (y_i, x)$ for $i = 1, 2$.
We then have two possibilities $\rmin(I + x - y_1 - y_2) = |I| - 1$ and $\rmin(I + x - y_1 - y_2) = |I| - 2$.

When $\rmin(I + x - y_1 - y_2) = |I| - 1$, at least one of $a_1$ and $a_2$ and at least one of $b_1$ and $b_2$ must be included in $\tD[I]$. %
This condition is simply represented by two clauses $\cC_{X, Y} \coloneqq \{C_{X, Y}^1, C_{X, Y}^2\}$ such that
\[C_{X, Y}^1 \coloneqq (z_{a_1} \vee z_{a_2}), \quad C_{X, Y}^2 \coloneqq (z_{b_1} \vee z_{b_2}).\]

When $\rmin(I + x - y_1 - y_2) = |I| - 2$, both $a_1$ and $a_2$ or both $b_1$ and $b_2$ must be excluded from $\tD[I]$. %
This condition is represented by four clauses
\begin{align*}
C_{X, Y}^1 \coloneqq (\neg z_{a_1} \vee \neg z_{b_2}), \ C_{X, Y}^2 \coloneqq (\neg z_{a_2} \vee \neg z_{b_1}), \
C_{X, Y}^3 \coloneqq (\neg z_{a_1} \vee \neg z_{b_1}), \ C_{X, Y}^4 \coloneqq (\neg z_{a_2} \vee \neg z_{b_2}).
\end{align*}
Since $\rmin(I + x - y_1 - y_2) = |I| - 2$ implies $\rmin(I + x - y_i) = |I| - 1$ for each $i \in \{1, 2\}$, the latter two clauses $C_{X, Y}^3$ and $C_{X, Y}^4$ have already been introduced as $C_{x, y_1}^1$ and $C_{x, y_2}^1$, respectively, in Case~1.1.
Thus, we define $\cC_{X, Y} \coloneqq \{C_{X, Y}^1, C_{X, Y}^2\}$, and we are done.

\subsubsection*{Case 2.1.~ When $|X| = 2$ and $|Y| = 1$.}
Suppose that $X = \{x_1, x_2\}$ and $Y = \{y\}$, and let $a_i = (x_i, y)$ and $b_i = (y, x_i)$ for $i = 1, 2$.
We then have two possibilities $\rmin(I + x_1 + x_2 - y) = |I|$ and $\rmin(I + x_1 + x_2 - y) = |I| - 1$.
Each case is analogous to Case~1.2.

When $\rmin(I + x_1 + x_2 - y) = |I|$, at least one of $a_1$ and $a_2$ and at least one of $b_1$ and $b_2$ must be included in $\tD[I]$.
This condition is simply represented by two clauses $\cC_{X, Y} \coloneqq \{C_{X, Y}^1, C_{X, Y}^2\}$ such that
\[C_{X, Y}^1 \coloneqq (z_{a_1} \vee z_{a_2}), \quad C_{X, Y}^2 \coloneqq (z_{b_1} \vee z_{b_2}).\]

When $\rmin(I + x_1 + x_2 - y) = |I| - 1$, both $a_1$ and $a_2$ or both $b_1$ and $b_2$ must be excluded from $\tD[I]$.
This condition is represented by four clauses
\begin{align*}
C_{X, Y}^1 \coloneqq (\neg z_{a_1} \vee \neg z_{b_2}), \ C_{X, Y}^2 \coloneqq (\neg z_{a_2} \vee \neg z_{b_1}),\
C_{X, Y}^3 \coloneqq (\neg z_{a_1} \vee \neg z_{b_1}), \ C_{X, Y}^4 \coloneqq (\neg z_{a_2} \vee \neg z_{b_2}).
\end{align*}
Since $\rmin(I + x_1 + x_2 - y) = |I| - 1$ implies $\rmin(I + x_i - y) = |I| - 1$ for each $i \in \{1, 2\}$, the latter two clauses $C_{X, Y}^3$ and $C_{X, Y}^4$ have already been introduced as $C_{x_1, y}^1$ and $C_{x_2, y}^1$, respectively, in Case~1.1.
Thus, we define $\cC_{X, Y} \coloneqq \{C_{X, Y}^1, C_{X, Y}^2\}$, and we are done.

\subsubsection*{Case 2.2.~ When $|X| = 2$ and $|Y| = 2$.}
Suppose that $X = \{x_1, x_2\}$ and $Y = \{y_1, y_2\}$, and let $a_{i,j} = (x_i, y_j)$ and $b_{i,j} = (y_j, x_i)$ for $i = 1, 2$ and $j = 1, 2$.
We then have three possible values $|I|$, $|I| - 1$, and $|I| - 2$ of $\rmin(I + x_1 + x_2 - y_1 - y_2)$.

When $\rmin(I + x_1 + x_2 - y_1 - y_2) = |I|$, at least one of $a_{i,j}$ and at least one of $b_{i,j}$ must be included in $\tD[I]$.
In this case, this is simply satisfied due to the clauses already introduced in Cases 1.2 and 2.1, since we have $\rmin(I + x_1 + x_2 - y_i) = |I|$ and $\rmin(I + x_i - y_1 - y_2) = |I| - 1$ for any $i \in \{1, 2\}$.
Thus, it suffices to set $\cC_{X, Y} = \emptyset$.

When $\rmin(I + x_1 + x_2 - y_1 - y_2) = |I| - 2$, all of $a_{i, j}$ or all of $b_{i, j}$ must be excluded from $\tD[I]$. %
In this case, we have $\rmin((I \cup X') \setminus Y') = |I| - |Y'|$ for any subpair $(X', Y')$ of $(X, Y)$.
The consistency with respect to the proper subpairs $(X', Y')$ implies (i.e., the clauses already introduced in Cases~1.1, 1.2, and 2.1 impose) that for each $i, j, k \in \{1, 2\}$ (possibly, $i = k$), at least one of $a_{i,j}$ and $b_{k,j}$ is excluded from $\tA[I]$ and at least one of $b_{i,j}$ and $a_{k,j}$ is excluded from $\tA[I]$.
Thus, the remaining possibility that we have to reject is that $\tD[I]$ has two arcs $a_{i, j}$ and $b_{3-i, 3-j}$ for some $i, j \in \{1,2\}$.
We then define $\cC_{X, Y} \coloneqq \{C_{X, Y}^1, C_{X, Y}^2, C_{X, Y}^3, C_{X, Y}^4\}$ with
\begin{align*}
C_{X, Y}^1 \coloneqq (\neg z_{a_{1,1}} \vee \neg z_{b_{2,2}}), \ C_{X, Y}^2 \coloneqq (\neg z_{a_{1,2}} \vee \neg z_{b_{2,1}}), \
C_{X, Y}^3 \coloneqq (\neg z_{a_{2,1}} \vee \neg z_{b_{1,2}}), \ C_{X, Y}^4 \coloneqq (\neg z_{a_{2,2}} \vee \neg z_{b_{1,1}}).
\end{align*}
It is easy to see that $\cC_{X, Y}$ is satisfied if and only if $z_{a_{i,j}} = \false$ or $z_{b_{3-i,3-j}} = \false$ for each $i, j \in \{1, 2\}$, which coincides with what we have to do.

When $\rmin(I + x_1 + x_2 - y_1 - y_2) = |I| - 1$, we separately consider two cases.
Suppose that $(X, Y)$ is not evil, i.e., $\rmin((I \cup X') \setminus Y') \ge |I| - |Y'| + 1$ for some proper subpair $(X', Y')$ of $(X, Y)$.
In this case, the consistency with respect to $(X', Y')$ implies (i.e., the clauses already introduced in Cases~1.1, 1.2, or 2.1 impose) the consistency with respect to $(X, Y)$.
Thus, it suffices to set $\cC_{X, Y} = \emptyset$.

Otherwise, $(X, Y)$ is evil, i.e., $\rmin((I \cup X') \setminus Y') = |I| - |Y'|$ for any proper subpair $(X', Y')$ of $(X, Y)$.
In this case, as seen in the second last paragraph, the consistency with respect to such LE-pairs $(X', Y')$ implies (i.e., the clauses already introduced in Cases~1.1, 1.2, and 2.1 impose) that for each $i, j, k \in \{1, 2\}$ (possibly, $i = k$), at least one of $a_{i,j}$ and $b_{k,j}$ is excluded from $\tA[I]$ and at least one of $b_{i,j}$ and $a_{k,j}$ is excluded from $\tA[I]$.
Thus, for the consistency with respect to $(X, Y)$, there must exist a pair of arcs $a_{i,j}$ and $b_{3-j,3-i}$ for exactly one pair $(i, j)$ with $i, j \in \{1, 2\}$ (cf.~Example~\ref{ex:LE-pair2}).
We then define $\cC_{X, Y} \coloneqq \{C_{X, Y}^1, C_{X, Y}^2, \dots, C_{X, Y}^8\}$ with
\begin{align*}
&C_{X, Y}^1 \coloneqq (\neg z_{a_{1,1}} \vee z_{b_{2,2}}), \ C_{X, Y}^2 \coloneqq (\neg z_{a_{1,2}} \vee z_{b_{2,1}}), \
C_{X, Y}^3 \coloneqq (\neg z_{a_{2,1}} \vee z_{b_{1,2}}), \ C_{X, Y}^4 \coloneqq (\neg z_{a_{2,2}} \vee z_{b_{1,1}}),\\
&C_{X, Y}^5 \coloneqq (z_{a_{1,1}} \vee \neg z_{b_{2,2}}), \ C_{X, Y}^6 \coloneqq (z_{a_{1,2}} \vee \neg z_{b_{2,1}}), \
C_{X, Y}^7 \coloneqq (z_{a_{2,1}} \vee \neg z_{b_{1,2}}), \ C_{X, Y}^8 \coloneqq (z_{a_{2,2}} \vee \neg z_{b_{1,1}}).
\end{align*}
We observe that $C_{X, Y}^1$ and $C_{X, Y}^5$ are both satisfied if and only if $z_{a_{1,1}} = z_{b_{2,2}}$, and the similar equivalence holds for the other three pairs of clauses.
Hence, $\cC_{X, Y}$ is satisfied in addition to the clauses already introduced in Cases~1.1, 1.2, and 2.1 if and only if $z_{a_{i,j}} = z_{b_{3-i,3-j}}$ for every $i, j \in \{1, 2\}$ and $z_{a_{i,j}} = z_{b_{3-i,3-j}} = \true$ for at most one pair $(i, j)$.
If $z_{a_{i,j}} = z_{b_{3-i,3-j}} = \true$ for some pair $(i, j)$, then $\tD[I]$ is consistent with respect to $(X, Y)$;
otherwise, all the eight variables $z_{a_{i,j}}, z_{b_{i,j}}$ appearing here take $\false$ and $\tD[I]$ is underestimated with respect to $(X, Y)$.

\medskip
Overall, by at most eight clauses each of which contains at most two variables, as required in Definition~\ref{def:almost_consistent_graph}, we can exactly represent the consistency with respect to every non-evil LE-pair, and we can impose that $\tD[I]$ is underestimated with respect to every evil LE-pair so that if $\tD[I]$ is not consistent with respect to an evil LE-pair $(X, Y)$, then $\tA[I]$ contains no edge between $X$ and $Y$.
This completes the proof.
\end{proof}

\subsection{Tractable Special Cases}\label{sec:trac}
We demonstrate several special cases that can be solved by this $2$-SAT underestimation approach.
Recall that $r$ denotes the maximum cardinality of a common independent set.

First, if no circuit in one matroid is included in any circuit in the other matroid, then one can solve the weighted matroid intersection problem in polynomial time only using the minimum rank oracle.
This special case already includes the cases of bipartite matching and arborescence (in simple graphs).

\begin{theorem}\label{thm:no_circuit_inclusion}
Suppose that for any pair of a circuit $C_1$ in $\mathbf{M}_1$ and a circuit $C_2$ in $\mathbf{M}_2$, we have $C_2 \setminus C_1 \neq \emptyset$.
Then, for any weight function $w \colon E \to \mathbb{R}$, a $w$-maximal common independent set in $\cImink$ for each cardinality $k = 0, 1, \dots, r$ can be found in $\mathrm{O}(r^3n^2)$ time only using the minimum rank oracle.
\end{theorem}

\begin{proof}
We obtain this result by combining Algorithm~\ref{alg:cheapest_path_augment_min_rank} and Lemma~\ref{lem:2-SAT}.
Specifically, under the assumption of this lemma, any almost consistent exchangeability graph $\tD[I]$ is in fact consistent.
Note that the computational time of finding a shortest cheapest $S_I^*$--$T_I^*$ path in $\tD[I]$ in Step~5 is $\mathrm{O}(rn^2)$ (using the Bellman--Ford algorithm), which is not a bottleneck.

To see the consistency of $\tD[I]$, it suffices to show that underestimation with respect to any evil LE-pair cannot occur.
Let $(X, Y)$ be an evil LE-pair.
As Case~2.2 in the proof of Lemma~\ref{lem:2-SAT}, suppose that $X = \{x_1, x_2\}$ and $Y = \{y_1, y_2\}$, and let $a_{i,j} = (x_i, y_j)$ and $b_{i,j} = (y_j, x_i)$ for $i = 1, 2$ and $j = 1, 2$.
We show that for at least one pair $(i, j)$, we must have $a_{i,j} \in \tA[I]$ or $b_{i,j} \in \tA[I]$, which implies that $\tD[I]$ is consistent with $(X, Y)$ (recall that the only possible underestimation is that all of the eight arcs between $X$ and $Y$ are excluded from $\tD[I]$).

Since $\rmin((I \cup X) \setminus Y) = |I| - 1 = |I| - |Y| + 1$, there exists at least one arc $(y_i, x_j) \in A[I]$ (in the true graph).
By the symmetry, we assume $(y_1, x_1) \in A[I]$.
This means that $y_1$ is in the fundamental circuit $C_1(I, x_1)$ of $x_1$ with respect to $I$ in $\mathbf{M}_1$.
By the assumption of the lemma, the fundamental circuit $C_2(I, x_1)$ in the other matroid $\mathbf{M}_2$ contains some element $y' \not\in C_1(I, x_1)$.

Consider an LE-pair $(X', Y') = (\{x_1\}, \{y_1, y'\})$.
Since $y' \in C_2(I, x_1) \setminus C_1(I, x_1)$, we have $\rmin(I + x_1 - y') = r_1(I + x_1 - y') = |I| - 1$.
Also, as $(X, Y)$ is evil, we have $\rmin(I + x_1 - y_1) = |I| - 1$.
Finally, since $y_1 \in C_1(I, x_1)$ and $y' \in C_2(I, x_1)$, we have $\rmin(I + x_1 - y_1 - y') = |I| - 1$.
Thus, the consistency with respect to the subpairs of $(X', Y')$ implies either $a_{1,1} \in \tA[I]$ or $b_{1,1} \in \tA[I]$, which completes the proof.
\end{proof}

\begin{remark}
    The condition of Theorem~\ref{thm:no_circuit_inclusion} raises the following problem: can we decide if there exists a circuit $C$ in $\mathbf{M}_1$ that is dependent in $\mathbf{M}_2$? This problem is hard even if an oracle access (or linear representation) is available for both matroids. To see this, let $G=(V,E)$ be an undirected graph, set $\mathbf{M}_1$ to be the graphic matroid of $G$ and $\mathbf{M}_2$ be the uniform matroid of rank $|V|-1$. Then, there exists a circuit $C$ in $\mathbf{M}_1$ that is dependent in $\mathbf{M}_2$ if and only if $G$ has a Hamiltonian cycle.   
\end{remark}

A similar discussion can lead to an FPT algorithm parameterized by the size of largest circuits.
This implies that the case when one matroid is a partition matroid with constant upper bounds is tractable.

\begin{theorem}\label{thm:FPT_circuit}
Let $\gamma_i$ be the size of largest circuits in $\mathbf{M}_i$, and $\gamma = \min\{\gamma_1, \gamma_2\}$.
Then, for any weight function $w \colon E \to \mathbb{R}$, a $w$-maximal common independent set in $\cImink$ for each cardinality  $k = 0, 1, \dots, r$ can be found in $\mathrm{O}(2^{\gamma}r^3n^2)$ time only using the minimum rank oracle.
\end{theorem}

\begin{proof}
Without loss of generality, we assume $\gamma = \gamma_2$.
Then, each $x \in E \setminus (I \cup T_I)$ has at most $\gamma$ outgoing arcs in $A_2[I]$.
In particular, by Lemma~\ref{lem:intersection}, there are at most $\gamma$ elements in $I$ that have arcs in $\Amint[I]$ from all the elements in $S_I$ (that are sure), and only those elements can have a suspicious arc in $\Asuspt[I]$.
Let $J \subseteq I$ be the set of such elements (that have a suspicious arc in $\Asuspt[I]$).
We then guess a subset $J' \subseteq J$ such that
\[J' = \{\, y \in J \mid \nexists (x, y) \in A_2[I] \setminus \Asuret[I]~\text{s.t.}~(y, x) \not\in A_1[I] \,\},\]
and show that a consistent exchangeability graph can be obtained by the $2$-SAT approach under the assumption that $J'$ is correctly guessed.
Since the number of the candidates of $J'$ is $2^{|J|} \le 2^\gamma$, this concludes the stated computational time.

Fix $J' \subseteq J$, and let $(X, Y)$ be an evil LE-pair.
As Case~2.2 in the proof of Lemma~\ref{lem:2-SAT}, suppose that $X = \{x_1, x_2\}$ and $Y = \{y_1, y_2\}$, and let $a_{i,j} = (x_i, y_j)$ and $b_{i,j} = (y_j, x_i)$ for $i = 1, 2$ and $j = 1, 2$.
We show that, by adding some clauses to $\cC_{X, Y}$ if necessary, we can impose the consistency with respect to $(X, Y)$ under the assumption that $J'$ is correctly guessed.

Suppose that $Y \setminus J \neq \emptyset$, say $y_1 \not\in J$.
Then, regardless of $J'$, we do the following.
If $\Amint[I]$ contains $a_{1,1}$ or $a_{2,1}$, then it is sure.
Thus, $(X, Y)$ cannot be underestimated under the original definition of $\cC_{X,Y}$, so we need no additional clause.
Otherwise, for the consistency with respect to $(X, Y)$, instead of $a_{1,1}$ and $a_{2,1}$, either $a_{1,2}$ or $a_{2,2}$ must be included in $\tA[I]$, which can be simply represented by an additional clause $(z_{a_{1,2}} \vee z_{a_{2,2}})$.

Otherwise, $Y \subseteq J$.
If $Y \subseteq J'$, then we can conclude that this $J'$ is a wrong guess, since any consistent exchangeability graph contains exactly one pair of $a_{i,j}$ and $b_{3-i,3-j}$ among the four candidates (cf.~Example~\ref{ex:LE-pair2}), so we skip this $J'$.
If $|Y \cap J'| = 1$, say $Y \cap J' = \{y_1\}$, then as with the previous paragraph, instead of $a_{1,1}$ and $a_{2,1}$ (that must be excluded in this trial), either $a_{1,2}$ or $a_{2,2}$ must be included in $\tA[I]$, which can be simply represented by an additional clause $(z_{a_{1,2}} \vee z_{a_{2,2}})$.
Otherwise, $Y \cap J' = \emptyset$.
In this case, we do not add any clauses, and the consistency with respect to $(X, Y)$ under the correct guess of $J'$ is seen as follows.

Suppose that the true graph $D[I]$ contains $a_{1,1}$ and $b_{2,2}$ (without loss of generality by the symmetry).
As $J'$ is correctly guessed, $D[I]$ contains some arc $(x', y_2) \in A_2[I] \setminus \Asuret[I]$ such that $(y_2, x') \not\in A_1[I]$.
Consider an LE-pair $(X', Y') = (\{x_2, x'\}, \{y_2\})$.
Then, as with the proof of Theorem~\ref{thm:no_circuit_inclusion}, we conclude that the consistency with respect to the subpairs of $(X', Y')$ implies either $a_{2,2} \in \tA[I]$ or $b_{2,2} \in \tA[I]$.
This completes the proof.
\end{proof}

\begin{remark}
Since the circuits of size at least $r + 1$ never appear in any augmentation step, we can sharpen Theorem~\ref{thm:FPT_circuit} by modifying the definition of $\gamma_i$ as the size of largest circuits in $\mathbf{M}_i$ among those of size at most $r$.
\end{remark}

Finally, a lexicographically maximal common independent set can be found in polynomial time by this approach, where a common independent set is \textbf{lexicographically maximal} if it takes as many heaviest elements as possible, and subject to this, it takes as many second heaviest elements as possible, and so on.

\begin{theorem}\label{thm:lexicographically_maximal}
A lexicographically maximal common independent set can be found in $\mathrm{O}(r^3n^2)$ time only using the minimum rank oracle.
\end{theorem}

\begin{proof}
Let $(E_1, E_2, \dots, E_\ell)$ be the partition of the ground set $E$ in the descending order of the weights, i.e., $E_1$ is the set of heaviest elements, $E_2$ is the set of second heaviest elements, and so on.
We show that any almost consistent exchangeability graph is sufficient (i.e., underestimation with respect to any evil LE-pair does not matter) to update a lexicographically maximal common independent set $I \in \cImink$ to a lexicographically maximal common independent set $I' \in \mathcal{I}_\mathrm{min}^{k+1}$ such that $|I' \cap E_i| \geq |I \cap E_i|$ for every $i = 1, 2, \dots, \ell$.
Indeed, if this is true, one can obtain a lexicographically maximal common independent set by starting with $I = \emptyset$, doing such updates in $E_1$ as long as possible, then doing so in $E_1 \cup E_2$, and so on.

Let us consider an auxiliary weighted matroid intersection instance obtained by replacing the weight function with $\hat{w} \colon E \to \mathbb{R}$ defined by $\hat{w}(e) \coloneqq (n + 1)^{\ell - i}$ for each $i = 1, 2, \dots, \ell$ and each $e \in E_i$.
It is easy to observe that a common independent set $I$ is $\hat{w}$-maximal in $\cImink$ if and only if $I$ is lexicographically maximal in $\cImink$. %
In what follows in this proof, we just say that $I$ is \textbf{maximal} to mean that $I$ is lexicographically maximal and $\hat{w}$-maximal.

Fix $k$, and let $I$ be a maximal common independent set in $\cImink$.
Suppose that $I$ can be augmented to a maximal common independent set $I' \in \mathcal{I}_\mathrm{min}^{k+1}$ such that $|I' \cap E_i| \geq |I \cap E_i|$ for every $i = 1, 2, \dots, \ell$.
This occurs if and only if the minimum cost of an $S_I$--$T_I$ path in $D[I]$ is negative, where the cost function $\hat{c} \colon E \to \mathbb{R}$ is defined analogously with respect to $\hat{w}$ (cf.~\eqref{eq:cost}).
Let $P$ be a shortest cheapest $S_I$--$T_I$ path in $D[I]$.
We show that this $P$ indeed exists as it is (including the order of elements) in any almost consistent exchangeability graph $\tD[I]$.
Since Lemma~\ref{lem:path_cycle_partition} implies that $\tD[I]$ (which is an underestimation) has no $S_I$--$T_I$ path $\tP$ such that either $\hat{c}(\tP) < \hat{c}(P)$ or $\hat{c}(\tP) = \hat{c}(P)$ and $|\tP| < |P|$, this completes the proof.

Suppose to the contrary that an almost consistent exchangeability graph $\tD[I] = (E \setminus I, I; \tA[I])$ does not contain $P$.
For the sake of simplicity, let $I' \coloneqq I \triangle P$ and assume $E = I \cup I'$.
We then have $I' \cap E_\ell \neq \emptyset$, and $|I' \cap E_i| = |I \cap E_i|$ for every $i < \ell$ by the maximality of $I \in \cImink$.
For each $i = 1, 2, \dots, \ell$, let $L_i \coloneqq \{\, e \in P \mid e \in E_j,\ j \le i \,\}$.
Then, $\emptyset \neq L_1 \subsetneq L_2 \subsetneq \cdots \subsetneq L_\ell = P$ and $|L_i|$ is even for every $i < \ell$.
The latter can be strengthened as follows, where a \textbf{segment} of $L_i$ is a maximal subset of $L_i$ whose elements appear consecutively on the path $P$, which is also regarded as the corresponding subpath of $P$.

\begin{claim}\label{cl:even}
For any $i < \ell$, each segment of $L_i$ is of even cardinality.
\end{claim}

\begin{proof}
Suppose to the contrary that for some $i < \ell$, there exists an odd segment $Z$.
Take such a counterexample so that $i$ is minimum and, among those, $Z$ appears on $P$ as early as possible.
We have either $|Z \cap I| = |Z \setminus I| - 1$ or $|Z \cap I| = |Z \setminus I| + 1$.

In the former case, let $x \in Z \setminus I$ be the last element of $Z$, and $y \in I$ be the successor of $x$ on $P$ (since $|L_i|$ is even and $Z$ is the first odd segment of $L_i$, $x$ is not the last element of $P$).
Let $P'$ be the prefix of $P$ ending at $y$, and $I'' \coloneqq I \triangle P'$.
Then, $I''$ is a common independent set with $|I''| = |I|$ by Lemma~\ref{lem:even_prefix}.
Also, by the choice of $i$ and $Z$, we have $|I'' \cap E_i| = |I \cap E_i| + 1$ and $|I'' \cap E_j| = |I \cap E_j|$ for every $j < i$, which contradicts the maximality of $I$.

In the latter case, let $y \in Z \cap I$ be the last element of $Z$ on $P$, and $y'$ be the next element in $I$ along $P$ (since $|L_i|$ is even, $Z$ is the first odd segment of $L_i$, and the successor of $y$ on $P$ is not in $L_i$, such $y'$ indeed exists).
Let $P'$ be the suffix of $P$ starting at $y'$, and $I'' \coloneqq I \triangle P'$.
Then, $I''$ is a common independent set with $|I''| = |I|$ by Lemma~\ref{lem:even_prefix}.
Also, by the choice of $i$ and $Z$, we have $|I'' \cap E_i| = |I \cap E_i| + 1$ and $|I'' \cap E_j| = |I \cap E_j|$ for every $j < i$, which contradicts the maximality of $I$.
\end{proof}

Since $\tD[I]$ does not contain the path $P$, there exists an arc $a = (x, y) \in P \setminus \tA[I]$.
Take such an arc $a$ so that the minimum index $i$ with $\{x, y\} \subseteq L_i$ is maximized.
By symmetry, we assume that $x \in E \setminus I$ and $y \in I$ (by interchanging $\mathbf{M}_1$ and $\mathbf{M}_2$ if necessary).
Since $a \in A[I] \setminus \tA[I]$, we have $a \in \Asusp[I]$ and hence $(s, y) \in A[I]$, where $s \in S_I$ is the first element of $P$.
We then see that there exists a special arc in the opposite direction on the prefix of $P$.

\begin{claim}\label{cl:mate}
There exists an arc $b = (y', x') \in \tA[I]$ on the prefix of $P$ ending at $x$ such that $y' \in I$, $x' \in E \setminus I$, $(x', y') \not\in \tA[I]$, $x' \in L_j$ and $y' \not\in L_j$ for some $j \ge i$, and the segment of $L_j$ starting with $x'$ contains $x$ and $y$.
\end{claim}

\begin{proof}
Since $L_\ell = P$, there exists $j \ge i$ such that the segment of $L_j$ containing $x$ and $y$ starts with some $x' \in E \setminus I$.
Take such an index $j$ as small as possible.
We then see that $x' \neq s$ as follows.

Suppose to the contrary that $x' = s$.
Then, the arc $(s, y) \in A[I]$ gives a shortcut for $P$ as follows.
Let $P'$ be the $S_I$--$T_I$ path obtained from $P$ by replacing its prefix ending at $y$ with the arc $(s, y)$.
Since the segment of $L_h$ containing $y$ starts in $I$ for any $i \le h < j$ by the choice of $j$, we have $\hat{c}(P') = \hat{c}(P)$ by Claim~\ref{cl:even}, which contradicts that $P$ is a shortest cheapest $S_I$--$T_I$ path in $D[I]$.

Thus, there exists an arc $b = (y', x') \in A[I]$ on the prefix of $P$ ending at $x$ such that $y' \in I$, $x' \in E \setminus I$, $x' \in L_j$, $y' \not\in L_j$, and the segment of $L_j$ starting with $x'$ contains $x$ and $y$.
Since $\hat{w}(x') > \hat{w}(y')$ (as $x' \in L_j$ and $y' \not\in L_j$), the reverse arc $(x', y')$ of $b$ is not contained in $A[I]$ (otherwise, $I + x' - y'$ is a heavier common independent set in $\cImink$, which contradicts the maximality of $I$).
By the choice of $a$ (the maximality of $i$), we have $b \in \tA[I]$, which completes the proof.
\end{proof}

Take an arc $b=(y',x')$ as in Claim~\ref{cl:mate} so that the index $j$ is minimized (see Figure~\ref{fig:lexico}). %
If $x' = x$, then the consistency with respect to the LE-pair $(\{x\}, \{y, y'\})$ implies $a \in \tA[I]$ (like Example~\ref{ex:LE-pair1}), a contradiction. %
Otherwise, we can see that the LE-pair $(X, Y) = (\{x, x'\}, \{y, y'\})$ is evil as follows, and hence the almost consistency with respect to $(X, Y)$ implies $a \in \tA[I]$, a contradiction again, which completes the proof.
Since the reverse arcs of $a = (x, y)$ and $b = (y', x')$ are simply excluded from $D[I]$ by the consistency with respect to the LE-pairs $(\{x\}, \{y\})$ and $(\{x'\}, \{y'\})$, it suffices to show the following claim (cf.~Example~\ref{ex:LE-pair2}).

\begin{figure}[t!]
\begin{center}
\includegraphics[width=\textwidth]{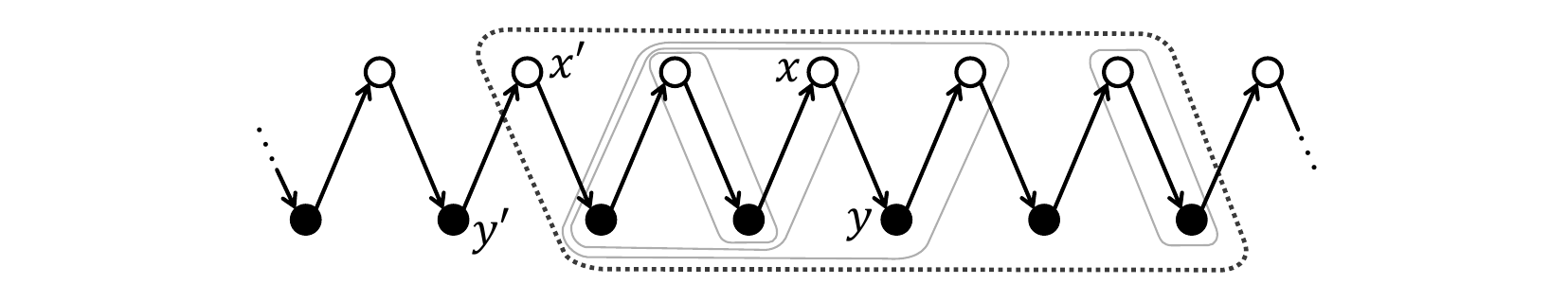}
\caption{A shortest cheapest $S_I$--$T_I$ path $P$, an arc $a = (x, y) \in P \setminus \tA[I]$, and an arc $b = (y',x')$ taken as in Claim~\ref{cl:mate} so that the index $j$ is minimized. %
Here, the segment of $L_j$ starting with $x'$ is enclosed by the dotted line. Other parts enclosed by solid lines are examples of segments of $L_{h}$ with $h < j$. %
}
\label{fig:lexico}
\end{center}
\end{figure}

\begin{claim}\label{cl:evil}
$A[I]$ contains none of $(x', y)$, $(y, x')$, $(x, y')$, and $(y', x)$.
\end{claim}

\begin{proof}
Suppose to the contrary that $(x', y) \in A[I]$.
Then, let $P'$ be the $S_I$--$T_I$ path obtained from $P$ by replacing its subpath from $x'$ to $y$ with the arc $(x', y)$.
By the choice of $b$ (the minimality of $j$) and Claim~\ref{cl:even}, we have $\hat{c}(P') = \hat{c}(P)$, which contradicts that $P$ is shortest among the cheapest $S_I$--$T_I$ paths in $D[I]$. 

Suppose to the contrary that $(y, x') \in A[I]$.
Since $(x, y) \in A[I]$, the consistency with respect to the LE-pair $(\{x, x'\}, \{y\})$ (together with the fact that $(x, y) \not\in \tA[I]$) implies that $(x', y), (y, x) \in \tA[I]$.
This contradicts the consistency with respect to the LE-pair $(\{x'\}, \{y, y'\})$, where $\rmin(I + x' - y - y') = |I| - 2$ since $A[I]$ contains neither $(x', y)$ nor $(x', y')$.

Suppose to the contrary that $(x, y') \in A[I]$.
Then, there exists a cycle $C$ in $D[I]$ whose vertex set coincides with the subpath of $P$ from $y'$ to $x$.
By the choice of $b = (y', x')$, where $x' \in L_j$ and $y' \not\in L_j$, and by Claim~\ref{cl:even}, we have $\hat{c}(C) < 0$, which contradicts that $I$ is maximal (cf.~Lemma~\ref{lem:negative_cycle}).

Suppose to the contrary that $(y', x) \in A[I]$.
Then, let $P'$ be the $S_I$--$T_I$ path obtained from $P$ by replacing its subpath from $y'$ to $x$ with the arc $(y', x)$.
By the choice of $b$ (the minimality of $j$ and the fact that the segment of $L_j$ starting with $x'$ contains both $x$ and $y$) and Claim~\ref{cl:even}, we have $\hat{c}(P') \le \hat{c}(P)$, which contradicts that $P$ is a shortest cheapest $S_I$--$T_I$ path in $D[I]$.
\end{proof}

To sum up, $(X, Y) = (\{x, x'\}, \{y, y'\})$ is an evil LE-pair, $a = (x, y) \not\in \tA[I]$, and $b = (y', x') \in \tA[I]$.
This contradicts that $\tD[I]$ is almost consistent, which completes the proof.
\end{proof}

By \cite{berczi2022approximation}, a lexicographically maximal solution gives $\min\{1, \alpha/2\}$-approximation of the maximum weight of a common independent set, where $\alpha > 1$ is the minimum ratio of two distinct positive weight values.
Note that when the objective is to find a maximum-weight common independent set, we can assume the weight function is positive by excluding any element of weight at most $0$ from the ground set in advance.

\begin{corollary}\label{cor:approximation}
There exists an $\mathrm{O}(r^3n^2)$-time $\min\{1, \alpha/2\}$-approximation algorithm for finding a maximum-weight common independent set only using the minimum rank oracle, where $\alpha > 1$ is the minimum ratio of two distinct positive weight values.
\end{corollary}

\section{NP-hardness of Finding a Consistent Exchangeability Graph}\label{sec:NP-hard}

In this section, we prove that the problem of finding a consistent exchangeability graph is NP-hard.
Recall that, by Lemma~\ref{lem:2-SAT}, an almost consistent exchangeability graph can be found in polynomial time via reduction to 2-SAT.

\let\tmp\thetheorem
\renewcommand{\thetheorem}{\ref*{thm:NP-hard}}
\begin{theorem}%
The following problem is NP-hard even if the two matroids $\mathbf{M}_1$ and $\mathbf{M}_2$ are both $\mathbb{Q}$-representable:
Given a common independent set $I$, the source and sink sets $S_I$ and $T_I$, and the values $\rmin((I \cup X) \setminus Y)$ for all the LE-pairs $(X, Y)$, find a consistent exchangeability graph $\tD[I]$.
\end{theorem}
\let\thetheorem\tmp
\addtocounter{theorem}{-1}

\begin{proof}
We reduce the following problem, which is known to be NP-hard~\cite{khanna2000hardness, guruswami2004hardness}.

\searchprob{4-coloring of 3-colorable graph}{A 3-colorable graph $G = (V, F)$.}{Find a 4-coloring of $G$.}

We first implicitly construct two matroids $\mathbf{M}_1$ and $\mathbf{M}_2$ on the same ground set $E$, a common independent set $I$, and the source and sink sets $S_I$ and $T_I$ by setting the values $\rmin((I \cup X) \setminus Y)$ for each LE-pair $(X, Y)$ so that there exists a natural correspondence between the consistent exchangeability graphs and the $4$-colorings of $G$. We then give representations of $\mathbf{M}_1$ and $\mathbf{M}_2$ over $\mathbb{Q}$ with respect to each consistent exchangeability graph.
The existence of a 3-coloring of $G$ implies that there exists a true graph $D[I]$ that is consistent, i.e., the constructed instance can indeed occur as the intersection of two $\mathbb{Q}$-representable matroids.
Thus, this completes the proof.

For each vertex $v \in V$, let us construct a vertex gadget consisting of four elements $x_1^v, x_2^v \in E \setminus (I \cup S_I \cup T_I)$ and $y_1^v, y_2^v \in I$, which form an evil LE-pair $(X^v, Y^v) = (\{x_1^v, x_2^v\}, \{y_1^v, y_2^v\})$, as follows.
Define $\rmin((I \cup X^v) \setminus Y^v) \coloneqq |I| - 1$ and $\rmin((I \cup X') \setminus Y') \coloneqq |I| - |Y'|$ for any proper subpair $(X', Y')$ of $(X^v, Y^v)$.
Let $a_{i,j}^v = (x_i^v, y_j^v)$ and $b_{i,j}^v = (y_{j}^v, x_{i}^v)$ for any $i, j \in \{1, 2\}$.
Then, the consistency with respect to the evil LE-pair $(X^v, Y^v)$ imposes that a pair of arcs $a_{i,j}^v$ and $b_{3-i,3-j}^v$ for exactly one pair $(i, j)$ of $i, j \in \{1, 2\}$ is included in $\tA[I]$ and the other three pairs are excluded (cf.~Example~\ref{ex:LE-pair2}).
These four possibilities correspond to which color is assigned to the vertex $v$ in a $4$-coloring of $G$.
For the sake of simplicity, let $\{(1, 1), (1, 2), (2, 1), (2, 2)\}$ be the set of four colors, and a color $(i, j)$ assigned to $v$ is identified with the situation that $a_{i,j}^v, b_{3-i,3-j}^v \in \tA[I]$ in the vertex gadget (see Figure~\ref{fig:hardness}).

For each edge $e = \{u, w\} \in F$, let us construct an edge gadget as follows.
We introduce six elements $x_1^{e,u}, x_1^{e,w}, x_2^{e,u}, x_2^{e,w} \in E \setminus (I \cup S_I \cup T_I)$ and $y_1^e, y_2^e \in I$.
For each $i \in \{1, 2\}$, let $(X_i^e, Y_i^e) = (\{x_i^{e,u}, x_i^{e,w}\}, \{y_i^e\})$.
Define $\rmin((I \cup X_i^e) \setminus Y_i^e) \coloneqq |I|$ and $\rmin(I + x_i^{e,u} - y_i^e) \coloneqq \rmin(I + x_i^{e,w} - y_i^e) \coloneqq |I| - 1$.
Then, the consistency with respect to the subpairs of $(X_i^e, Y_i^e)$ imposes that one of $\{(x_i^{e,u}, y_i^e), (y_i^e, x_i^{e,w})\}$ and $\{(x_i^{e,w}, y_i^e), (y_i^e, x_i^{e,u})\}$ is a subset of $\tA[I]$ and the other is disjoint from $\tA[I]$ (see Figure~\ref{fig:hardness}).

\begin{figure}[t!]
\begin{center}
\includegraphics[width=\textwidth]{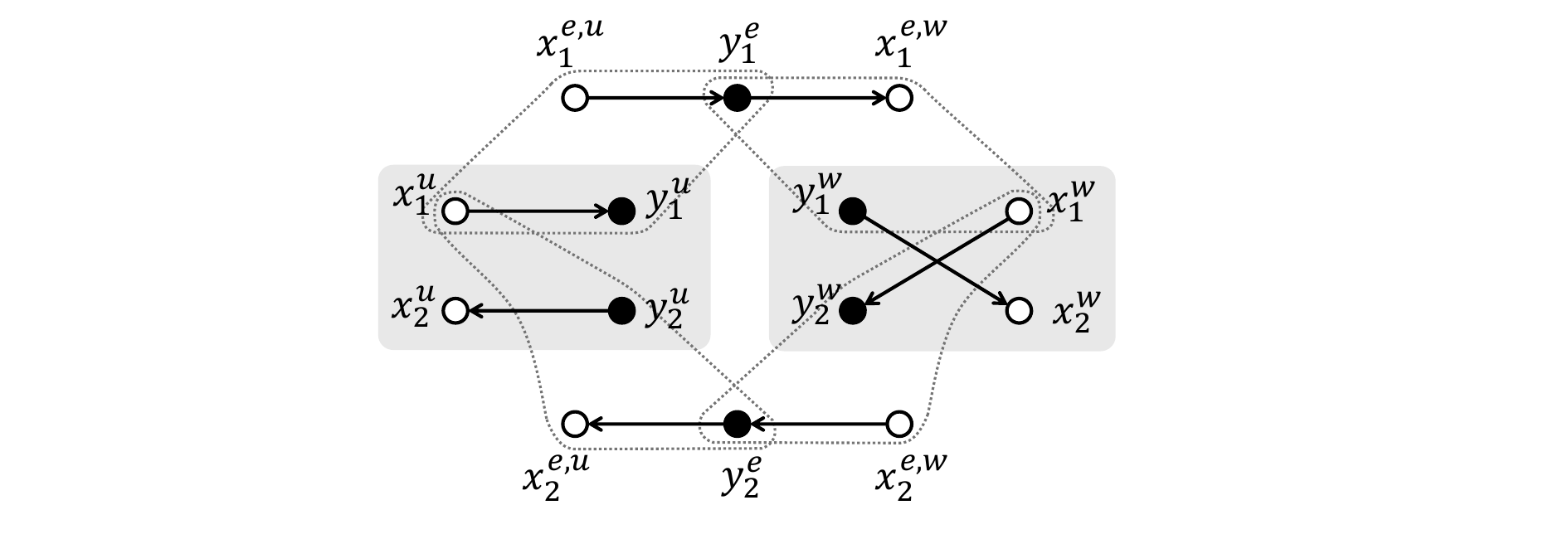}
\caption{The vertex gadgets for $u, w \in V$, together with the edge gadget for $e = \{u, w\} \in F$.
White and black vertices represent elements in $E\setminus (I\cup S_I\cup T_I)$ and in $I$, respectively.
The two parts with gray background are two evil LE-pairs $(X^u,Y^u)$ and $(X^w,Y^w)$ in the vertex gadgets, and the four parts enclosed by dotted lines are the LE-pairs $(X_i^{e,v}, Y_i^{e,v})$ $(i \in \{1, 2\},~v \in \{u, w\})$ defined in the edge gadget.
The illustrated realization of exchangeability arcs corresponds to assigning color $(1,1)$ to $u$ and $(1,2)$ to $w$, where we are omitting all pairs of opposite directed arcs $(x, y)$ and $(y, x)$ that actually exist between any pair of white and black vertices $x$ and $y$, respectively, not included in any LE-pairs mentioned above (e.g., $x = x_i^u$ and $y = y_j^w$ $(i, j \in \{1, 2\})$, $x = x_1^{e,v}$ and $y = y_2^{v'}$ $(v, v' \in \{u, w\})$, or $x = x_2^v$ and $y = y_i^e$ $(v \in \{u, w\},~i \in \{1, 2\})$).}
\label{fig:hardness}
\end{center}
\end{figure}

In addition, let
\begin{align*}
(X_1^{e,u}, Y_1^{e,u}) = (\{x_1^{e,u}, x_1^u\},\, \{y_1^e, y_1^u\}),\quad (X_1^{e,w}, Y_1^{e,w}) = (\{x_1^{e,w}, x_1^w\},\, \{y_1^e, y_1^w\}),\\
(X_2^{e,u}, Y_2^{e,u}) = (\{x_2^{e,u}, x_1^u\},\, \{y_2^e, y_2^u\}),\quad (X_2^{e,w}, Y_2^{e,w}) = (\{x_2^{e,w}, x_1^w\},\, \{y_2^e, y_2^w\}),
\end{align*}
and define $\rmin((I \cup X') \setminus Y') \coloneqq |I| - |Y'|$ for each subpair $(X', Y')$ of $(X_i^{e,v}, Y_i^{e,v})$ $(i \in \{1, 2\},\ v \in \{u, w\})$.
Then, the consistency with respect to the subpairs of $(X_i^{e,v}, Y_i^{e,v})$ imposes that $u$ and $w$ are assigned different colors (see Figure~\ref{fig:hardness}).
Suppose to the contrary that, for example, $u$ and $w$ are assigned the same color $(1, 1)$, i.e., $a_{1,1}^u = (x_1^u, y_1^u) \in \tA[I]$ and $a_{1,1}^w = (x_1^w, y_1^w) \in \tA[I]$.
Then, the consistency with respect to $(X_1^{e,u}, Y_1^{e,u})$ imposes that $(y_1^e, x_1^{e,u})$ is excluded from $\tA[I]$, and the consistency with respect to $(X_1^{e,w}, Y_1^{e,w})$ imposes that $(y_1^e, x_1^{e,w})$ is excluded from $\tA[I]$ as well.
This, however, contradicts the consistency with respect to $(X_1^e, Y_1^e)$, which imposes that either $\{(x_1^{e,u}, y_1^e), (y_1^e, x_1^{e,w})\}$ or $\{(x_1^{e,w}, y_1^e), (y_1^e, x_1^{e,u})\}$ must be included in $\tA[I]$.
The collision of the color $(2, 2)$ is forbidden in almost the same way, and the collision of $(1, 2)$ or $(2, 1)$ is similarly forbidden by the consistency with respect to $(X_2^{e,u}, Y_2^{e,u})$, $(X_2^{e,w}, Y_2^{e,w})$, and $(X_2^e, Y_2^e)$.
Thus, this edge gadget represents that $u$ and $w$ are assigned different colors.

As the last step of the construction, we put two special elements $s, t \in E \setminus I$ as a unique source in $S_I$ and a unique sink in $T_I$, respectively;
i.e., $\rmin(I + s) = \rmin(I + t) = |I|$ and $\rmin(I + s + t) = |I| + 1$.
For each pair of $x \in E \setminus (I \cup S_I \cup T_I)$ and $y \in I$, we define 
\[\rmin(I + s - y) \coloneqq \rmin(I + t - y) \coloneqq \rmin(I + s + x - y) \coloneqq \rmin(I + t + x - y) \coloneqq |I|,\]
which means that $C_2(I, s) = I + s$ and $C_1(I, t) = I + t$.
Also, for each LE-pair $(X, Y)$ that we have not yet defined the value $\rmin((I \cup X) \setminus Y)$, we define it as follows.
\begin{itemize}
\item If $|X| = 1$ or $|Y| = 1$, then $\rmin((I \cup X) \setminus Y) \coloneqq |I| - |Y| + 1$.
\item If $|X| = |Y| = 2$, then
\begin{itemize}
    \item $\rmin((I \cup X) \setminus Y) \coloneqq |I| - 1$ if there exists a proper subpair $(X', Y')$ of $(X, Y)$ such that $|X'| + |Y'| = 3$ and $\rmin((I \cup X') \setminus Y') = |I| - |Y'|$ (such $(X', Y')$ is a proper subpair of either $(X^v, Y^v)$ for some $v \in V$ or $(X_i^{e,v}, Y_i^{e,v})$ for some $i \in \{1, 2\}$ and some $v \in e \in E$), and
    \item $\rmin((I \cup X) \setminus Y) \coloneqq |I|$ otherwise (in this case, $X \cup Y$ has a cycle of length $4$ or can be partitioned into two disjoint cycles of length $2$).
\end{itemize}
\end{itemize}
We then see that
\begin{itemize}
\item all the pairs $(u, v)$ with either $u \in E \setminus (I \cup S_I \cup T_I)$ and $v \in I$ or $u \in I$ and $v \in E \setminus (I \cup S_I \cup T_I)$ are suspicious arcs (recall Lemma~\ref{lem:intersection} and the definition of suspicious arcs), and
\item every LE-pair $(X, Y)$ not designated in the construction of the vertex and edge gadgets contains a pair of $x \in X$ and $y \in Y$ exchangeable in both matroids, which can be determined from $\rmin(I + x - y) = |I|$, and then induces no more correlation between two suspicious arcs in the gadgets (recall Definition~\ref{def:consistent}).
\end{itemize}
Thus, we essentially have uncertainty only in the vertex gadgets and the edge gadgets, which are constrained so that each vertex gadget has four possible situations corresponding to the four colors, and that two vertex gadgets adjacent by an edge gadget are in different situations corresponding to the validity of $4$-coloring.
This gives a one-to-one correspondence between the situations of the vertex gadgets in the consistent exchangeability graphs and the $4$-colorings of $G$. Note that there may be multiple possible situations of each edge gadget for the same situations of the vertex gadgets (e.g., for an edge $e = \{u, w\} \in F$, if $u$ and $w$ are assigned color $(1, 1)$ and $(2, 2)$, then either of $\{(x_2^{e, u}, y_2^e), (y_2^e, x_2^{e, w})\}$ and $\{(x_2^{e, w}, y_2^e), (y_2^e, x_2^{e, u})\}$ can be included in $\tA[I]$).

Finally, we give representations of $\mathbf{M}_1$ and $\mathbf{M}_2$ over $\mathbb{Q}$ with respect to each consistent exchangeability graph, which indeed satisfy the minimum rank values defined above.
Fix a consistent exchangeability graph $\tD[I] = (E \setminus I, I; \tA[I])$.
We construct the representation matrices $Z_1$ and $Z_2$ of $\mathbf{M}_1$ and $\mathbf{M}_2$, respectively, as follows.
For each $i = 1, 2$, the row and column sets of $Z_i$ correspond to $I \cup \{s, t\}$ and $E$, respectively.
We denote the $(y, x)$ entry of $Z_i$ by $Z_i[y, x]$, and the $(Y, X)$ submatrix of $Z_i$ by $Z_i[Y, X]$.
As usual, each subset $X \subseteq E$ corresponds to the submatrix $Z_i[I \cup \{s, t\}, X]$ in the sense that $r_i(X) = \mathrm{rank}(Z_i[I \cup \{s, t\}, X])$.

We start with defining $Z_i[I, I]$ as the identity matrix with $Z_i[y, y] = 1$ $(\forall y \in I)$ and $Z_i[s, y] \coloneqq Z_i[t, y] \coloneqq 0$ for each $y \in I$.
We then define $Z_1[s, s] \coloneqq Z_2[t, t] \coloneqq 1$ and the other entries of $Z_i[\{s, t\}, \{s, t\}]$ as $0$.
Also, $Z_1[y, t] \coloneqq Z_2[y, s] \coloneqq 1$ and $Z_1[y, s] \coloneqq Z_2[y, t] \coloneqq 0$ for each $y \in I$ (recall that $C_2(I, s) = I + s$ and $C_1(I, t) = I + t$).
For the remaining entries, we fix mutually distinct prime numbers $p_{x,y} \in \mathbb{Q}$ for the pairs of $x \in E \setminus I$ and $y \in I$, and define $Z_i[y, x]$ as either $p_{x,y}$ or $0$ as follows, so that any $2 \times 2$ submatrix of $Z_i$ having nonzero values at both diagonal entries or both nondiagonal entries always becomes nonsingular:
\[Z_1[y, x] \coloneqq
\begin{cases}
p_{x,y} & \text{if $(y, x) \in \tA_1[I]$},\\
0 & \text{otherwise},
\end{cases} \qquad
Z_2[y, x] \coloneqq
\begin{cases}
p_{x,y} & \text{if $(x, y) \in \tA_2[I]$},\\
0 & \text{otherwise}.
\end{cases}
\]
Then, it is not difficult to check that the constructed matrices indeed result in two matroids satisfying the minimum rank values defined above, which completes the proof.
\end{proof}

\begin{remark}
The above reduction does not work for binary matroids, because the simultaneous exchangeability (the nonsingularity of $2 \times 2$ submatrices of $Z_i$) is sometimes broken, thus providing additional information.
The tractability of finding a consistent exchangeability graph for special classes of matroids (e.g., regular, binary, and $\mathbb{F}$-representable for some finite field $\mathbb{F}$) is an interesting open problem. By Lemma~\ref{lem:correctness_consistent}, an algorithm for finding consistent exchangeability graphs for a class would immediately lead to the tractability of weighted matroid intersection in that class under the minimum rank oracle.
\end{remark}

\begin{remark}
In the proof of Theorem~\ref{thm:NP-hard}, we only claim that all the consistent exchangeability graphs obtained in the way described there cannot be distinguished by the minimum rank values $\rmin((I \cup X) \setminus Y)$ for the LE-pairs $(X, Y)$. The result might be strengthened by showing that the same holds even if we observe $\rmin(Z)$ for all the subsets $Z \subseteq E$. Such a strengthening would indicate the limitation of any approach that relies on consistent exchangeability graphs.
\end{remark}

\section{Polymatroid Intersection is Hard under Minimum Rank Oracle}\label{sec:polymatroid}

In this section, we show the hardness of the unweighted polymatroid intersection under the minimum rank oracle.
A \textbf{polymatroid function} is a set function $\rho\colon 2^E \to \Rp$ that satisfies the following:
\begin{enumerate}[\rm (PM1)]
    \item $\rho(\emptyset) = 0$.
    \item monotonicity: $\rho(X) \le \rho(Y)$ for $X \subseteq Y \subseteq E$.
    \item submodularity: $\rho(X) + \rho(Y) \ge \rho(X \cup Y) + \rho(X \cap Y)$ for any $X, Y \subseteq E$.
\end{enumerate}
A polymatroid function is \textbf{integer} if $\rho(X) \in \Z$ for any $X \subseteq E$.
For $d \in \Zp$, a \textbf{$d$-polymatroid function} is an integer polymatroid function $\rho$ satisfying $\rho(\set{e}) \le d$ for every $e \in E$; 1-polymatroid functions are nothing but the rank functions of matroids.
A polymatroid function $\rho$ defines a polyhedron
\begin{align*}
    \mathrm{P}(\rho) \coloneqq \Set[\big]{x \in \Rp^E}{\text{$x(X) \le \rho(X)$ for all $X \subseteq E$}},
\end{align*}
which is referred to as the \textbf{polymatroid} defined by $\rho$.
The polymatroids defined by integral polymatroid functions are integral polyhedra~\cite{edmonds1970submodular}.

The (unweighted) \textbf{polymatroid intersection} problem is defined as follows:
given two polymatroids defined by $\rho_1, \rho_2\colon 2^E \to \Rp$, find a common point
$x \in \mathrm{P}(\rho_1) \cap \mathrm{P}(\rho_2)$ that maximizes $x(E)$.
As in the case of matroids, the intersection
of two integral polymatroids is an integral polyhedron~\cite{edmonds1970submodular}, and a polynomial-time algorithm
exists when the rank functions $\rho_1$ and $\rho_2$ are given separately
(see~\cite[Section~47.1]{schrijver2003combinatorial}).
Given Theorem~\ref{thm:unweighted}, it is natural to expect that the polymatroid intersection
remains tractable under the minimum rank oracle
$\rho_{\mathrm{min}}(X) \coloneqq \min\set{\rho_1(X), \rho_2(X)}$.
Unfortunately, this fails even if the input is restricted to instances where one polymatroid
is a matroid and the other is a $2$-polymatroid.

\begin{theorem}\label{thm:polymatroid}
    The unweighted polymatroid intersection problem for two polymatroids given by a minimum rank oracle requires an exponential number of oracle calls, even when one polymatroid is a matroid rank function and the other is a 2-polymatroid function.
\end{theorem}

\begin{proof}%
Let $\rho_1 \colon 2^E \to \Zp$ be the rank function of the free matroid, i.e., $\rho_1(X) \coloneqq |X|$ for $X \subseteq E$.
Given $T \subseteq E$, we define $\rho_2^T\colon 2^E \to \Zp$ as $\rho_2^T \coloneqq \rho_1$ if $T = \emptyset$ and
\begin{align}
    \rho_2^T(X) &\coloneqq \begin{cases}
        \abs{X} - 1 & (X = T), \\
        \abs{X} & (\text{$X \subsetneq T$ or $X \supsetneq T$}), \\
        \abs{X} + 1 & (\text{otherwise})
    \end{cases}
\end{align}
if $T \ne \emptyset$.

\begin{lemma}
    For any $T \subseteq E$, $\rho_2^T$ is a 2-polymatroid function.
\end{lemma}

\begin{proof}
    Assume $T \ne \emptyset$.
    First, $\rho_2^T$ is obviously integer-valued, and $\rho_2^T(\emptyset) = |\emptyset| = 0$ as $\emptyset \subsetneq T$.
    Second, $\rho_2^T(\set{e}) \le |\{e\}| + 1 = 2$ for every $e \in E$.
    Third, $\rho_2^T$ is indeed monotone; otherwise, there exist $X \subseteq E$ and $e \in E \setminus X$ such that $\rho_2^T(X) = |X| + 1 > |X| = |X + e| - 1 = \rho_2^T(X + e)$, but in this case $X + e = T$ and hence $X \subsetneq T$, a contradiction.
    
    Finally, we confirm the submodularity of $\rho_2^T$ as follows.
    Take any subsets $X, Y \subseteq T$ such that $X \not\subseteq Y$ and $X \not\supseteq Y$ (otherwise, the inequality is trivial).
    If $X = T$, then $T \not\subseteq Y \not\subseteq T$, $X \cup Y \supsetneq T$, and $X \cap T \subsetneq T$, which implies
    \[\rho_2^T(X) + \rho_2^T(Y) = |X| - 1 + |Y| + 1 = |X \cup Y| + |X \cap Y| = \rho_2^T(X \cup Y) + \rho_2^T(X \cap Y).\]
    By symmetry, suppose that $X \neq T \neq Y$.
    If $\rho_2^T(X) = |X|$ and $\rho_2^T(Y) = |Y|$, then either both $X$ and $Y$ are included in $T$ or both include $T$ (otherwise, $X \subsetneq T \subsetneq Y$ or $X \supsetneq T \supsetneq Y$, a contradiction).
    In this case, both $X \cup Y$ and $X \cap Y$ are in the same situation (one of them may be exactly $T$), respectively, and hence
    \[\rho_2^T(X) + \rho_2^T(Y) = |X| + |Y| = |X \cup Y| + |X \cap Y| \ge \rho_2^T(X \cup Y) + \rho_2^T(X \cap Y).\]
    Suppose that $\rho_2^T(X) = |X|$ and $\rho_2^T(Y) = |Y| + 1$.
    In this case, similarly, $X \cup Y$ or $X \cap Y$ is in the same situation as $X$ (depending on the situation of $X$), and hence
    \[\rho_2^T(X) + \rho_2^T(Y) = |X| + |Y| + 1 = |X \cup Y| + |X \cap Y| + 1 \ge \rho_2^T(X \cup Y) + \rho_2^T(X \cap Y).\]
    If $\rho_2^T(X) = |X| + 1$ and $\rho_2^T(Y) = |Y| + 1$, then the inequality is trivial.
    Thus, we complete the proof.
\end{proof}

Let $P^T \coloneqq \mathrm{P}(\rho_1) \cap \mathrm{P}(\rho_2^T)$.
Since
\begin{align}
    \rhomin^T(X)
    \coloneqq \min\set{\rho_1(X), \rho_2^T(X)}
    = \begin{cases}
        |X| - 1 & (X = T \ne \emptyset), \\
        |X| & (\text{otherwise}),
    \end{cases}
\end{align}
we have $P^T \subseteq {[0, 1]}^E$.
Furthermore, $\ones \in \R^E$ is in $P^T$ if and only if $T = \emptyset$.
Therefore, the optimal value of the unweighted polymatroid intersection problem for $\rho_1$ and $\rho_2^T$ is $|E|$ if $T = \emptyset$ and $|E| - 1$ if $T \ne \emptyset$.
Only the way to verify whether $T = \emptyset$ or not is to check $\rhomin^T(X) = |X|$ for all $X \subseteq E$, requiring exponentially many oracle calls to $\rhomin$.
\end{proof}

\paragraph{Acknowledgment.}
We are grateful to Yuni Iwamasa for the initial discussion on the problem, and to Ryuhei Mizutani for a discussion on the relation to $1/3$- and $2/3$-submodular functions.

Taihei Oki was supported by JST FOREST Grant Number JPMJFR232L and JSPS KAKENHI Grant Number JP22K17853. Yutaro Yamaguchi was supported by JSPS KAKENHI Grant Numbers 20K19743 and 20H00605, by Overseas Research Program in Graduate School of Information Science and Technology, Osaka University, and by the Osaka University Research
Abroad Program. Yu Yokoi was supported JST PRESTO Grant Number JPMJPR212B and JST ERATO Grant Number JPMJER2301.
The research was supported by the Lend\"ulet Programme of the Hungarian Academy of Sciences -- grant number LP2021-1/2021, by the Ministry of Innovation and Technology of Hungary from the National Research, Development and Innovation Fund -- grant numbers ADVANCED 150556 and ELTE TKP 2021-NKTA-62, by Dynasnet European Research Council Synergy project -- grant number ERC-2018-SYG 810115.

\bibliographystyle{abbrv}
\bibliography{min_rank}

\end{document}